\documentclass[11pt,letterpaper]{article}

\usepackage[margin=1in]{geometry}

\usepackage[utf8]{inputenc}
\usepackage[english]{babel}

\usepackage[dvipsnames,usenames]{xcolor}
\usepackage[colorlinks=true,pdfpagemode=UseNone,urlcolor=RoyalBlue,linkcolor=RoyalBlue,citecolor=OliveGreen,pdfstartview=FitH]{hyperref}

\usepackage{subcaption}

\usepackage{mathtools}
\usepackage{amsthm,thmtools}
\usepackage{amsmath}
\usepackage{amsfonts}
\usepackage{bm}

\usepackage[nameinlink]{cleveref}

\usepackage{tikz}
\usepackage{pgfplots}
\pgfplotsset{compat=newest}
\usetikzlibrary{shapes, shapes.symbols, backgrounds, matrix, calc, arrows, math, arrows.meta, decorations.pathreplacing, decorations.markings, patterns, calligraphy}

\usepackage{booktabs}

\usepackage{nicefrac}

\usepackage{enumitem}
\usepackage[ruled]{algorithm2e} 
\SetAlFnt{\small}
\SetAlCapFnt{\small}
\SetAlCapNameFnt{\small}
\SetAlCapHSkip{0pt}
\IncMargin{-\parindent}

\providecommand{\email}[1]{\href{mailto:#1}{\nolinkurl{#1}\xspace}}

\newcommand{\TV}{\text{\rm TV}}
\newcommand{\Kol}{\text{\rm Kol}}
\newcommand{\E}{\mathbf{E}}
\renewcommand{\Pr}{\mathbf{Pr}}

\DeclareMathOperator*{\argmax}{arg\,max}
\newcommand{\OPT}{\text{\rm OPT}}

\DeclareMathOperator*{\lin}{LIN}
\newcommand{\principalUtility}{u_{\mathrm{p}}}
\newcommand{\principalUtilityE}{\tilde{u}_{\mathrm{p}}}
\newcommand{\agentUtility}{u_{\mathrm{a}}}

\newcommand{\calA}{\mathcal{A}}
\newcommand{\calB}{\mathcal{B}}

\newtheorem{proposition}{Proposition}[section]
\newtheorem{lemma}[proposition]{Lemma}
\newtheorem{theorem}[proposition]{Theorem}
\newtheorem*{theorem*}{Theorem}
\newtheorem{corollary}[proposition]{Corollary}

\newtheorem{definition}{Definition}

\newtheorem{assumption}{Assumption}

\Crefname{equation}{Eq.}{Eqs.}

\usepackage{mdframed}
\newmdenv[skipabove=2mm, skipbelow=2mm, backgroundcolor=black!10]{boxedtext*}
\newmdenv[skipabove=2mm, skipbelow=2mm, rightline=true, leftline=true]{boxedtext}

\title{Are Bounded Contracts Learnable and Approximately Optimal?}

\author{Yurong Chen \\Peking University \\\email{chenyurong@pku.edu.cn} \and  Zhaohua Chen \\Peking University \\ \email{chenzhaohua@pku.edu.cn} \and Xiaotie Deng\\ Peking University \\\email{xiaotie@pku.edu.cn} \and Zhiyi Huang \\The University of Hong Kong\\ \email{zhiyi@cs.hku.hk}}

\date{}

\begin{document}

\begin{titlepage}
\thispagestyle{empty}
\maketitle
\begin{abstract}
    \thispagestyle{empty}
    This paper considers the hidden-action model of the principal-agent problem, in which a principal incentivizes an agent to work on a project using a contract.
    We investigate whether contracts with bounded payments are learnable and approximately optimal.
    Our main results are two learning algorithms that can find a nearly optimal bounded contract using a polynomial number of queries, under two standard assumptions in the literature:
    a costlier action for the agent leads to a better outcome distribution for the principal, and the agent's cost/effort has diminishing returns.
    Our polynomial query complexity upper bound shows that standard assumptions are sufficient for achieving an exponential improvement upon the known lower bound for general instances.
    Unlike the existing algorithms which relied on discretizing the contract space, our algorithms directly learn the underlying outcome distributions. 
    As for the approximate optimality of bounded contracts, we find that they could be far from optimal in terms of multiplicative or additive approximation, but satisfy a notion of mixed approximation.
\end{abstract}
 \end{titlepage}

\thispagestyle{empty}
\tableofcontents
\newpage
\setcounter{page}{1}

\section{Introduction}
\label{sec:introduction}

Contract theory is considered one of the pillars of economic theory~\cite{dutting2019simple}.
A fundamental model of contract design is the \emph{hidden-action} or \emph{moral hazard} model~\cite{mas1995microeconomic}.
Here, we consider a principal (e.g., a firm) and an agent (e.g., a worker).
The principal delegates a project to the agent, who takes an action that determines the distribution of the project's outcome.
The principal can only observe the outcome but not the agent's action.
Hence, it incentivizes the agent through a \emph{contract}, which specifies the payment to the agent for each outcome.
The principal aims to maximize the expectation of its value for the outcome minus the payment to the agent.
This classical model has been extensively studied~\cite{innes1990limited, grossman1992analysis, baker2000use} and enjoys a wide range of applications, including health insurance~\cite{lakdawalla2013health}, crowdsourcing~\cite{ho2015incentivizing}, and smart contracts~\cite{sheth2020blockchain}.

In the past few years, contract theory has been re-examined through the algorithmic lens in the algorithmic game theory literature (e.g., \cite{dutting2019simple, guruganesh2021contracts, dutting2021the}).
This paper focuses on learning a nearly optimal contract from incomplete information on a contract design instance.
The above classical model assumes that we have full information on the agent's set of actions and the corresponding outcome distributions.
In reality, these distributions are merely a modeling ingredient for capturing all the (incomplete) information that we have about the relationship between the agent's action and the project's outcome;
hence, we would never have exact knowledge of them.

Several recent works investigated learning the optimal contract in the bandit model~\cite{ho2016adaptive, cohen2022learning, zhu2023the}.\footnote{The model in \cite{ho2016adaptive} and \cite{zhu2023the} further lets the agent draw a type from an unknown distribution. Nonetheless, the hardness by \cite{zhu2023the} arises from the fixed-type case, which we focus on in this paper. It is an interesting future direction to extend our results to consider a distribution of the agent's types.}
In this model, the principal repeatedly chooses a contract for the same contract design instance for $T$ times, initially with no information about the outcome distributions.
In each round, the principal selects a contract based on the history of all previous rounds. Then, the agent takes an action that maximizes its utility under this contract, but the principal cannot observe the action.
Finally, a project outcome, observable to the principal, is drawn from the distribution that corresponds to the agent's action.
The goal is to minimize the \emph{regret}, which is the gap between the principal's cumulative utility and what could have been achieved if the optimal contract had been used from the beginning.
A \emph{no-regret} learning algorithm ensures that the regret is at most $o(T)$.

\cite{ho2016adaptive} studied \emph{bounded} and \emph{monotone} contracts and presented a no-regret algorithm for this class of contracts. 
Normalizing the principal's values for the project outcomes to be at most $1$, boundedness refers to the condition that the payment to the agent is also at most $1$.
Monotonicity says that an outcome with a higher principal's value will have a higher payment.
\cite{cohen2022learning} extended the no-regret learning result to risk-averse agents, focusing on bounded and \emph{smooth-monotone} contracts.
Smooth-monotonicity is stronger than monotonicity,
as it further 
requires
that an outcome with a higher principal's value 
also has
 a higher principal's utility, which equals the difference between the value and the payment.
Recently, \cite{zhu2023the} gave a no-regret learning algorithm that does not require monotonicity or smooth-monotonicity.
However, they still assume boundedness.
They also provided a lower bound for the best achievable regret.

\subsection{Our Results}

This work is motivated by two research gaps that have not been addressed in existing studies on learning optimal bounded contracts. 
First, the regret upper bounds by existing no-regret algorithms are asymptotically very close to $T$.
In fact, they are larger than $T$ unless $T$ is exponential in $m$, the number of possible project outcomes.
For example, the regret upper bound by \cite{zhu2023the} was $\tilde{O}(\sqrt{m}\cdot T^{\nicefrac{2m}{2m+1}})$.
If this is smaller than $T$, then $T$ must be exponentially large so that $T^{\nicefrac{1}{2m+1}}$ dominates the logarithmic factors in the $\tilde{O}$ notation.
Although their lower bound showed that having an exponentially large $T$ is unavoidable for an arbitrary contract design problem, it points to an obvious open question:
\begin{quote}
    \emph{Could we get $O(T^c)$ regret for some constant $c < 1$ under mild assumptions that are already considered in the literature?}
\end{quote}

Second, all previous studies have focused on bounded contracts due to the difficulty in learning unbounded contracts.
However, they have not analyzed the approximate optimality of this specific class of contracts.
Therefore, it is natural to ask:

\begin{quote}
    \emph{How much is left on the table when we restrict ourselves to bounded contracts?}
\end{quote}

\subsubsection*{Learnability of Bounded Contracts}

We consider two assumptions in the literature, known as First-Order Stochastic Dominance (FOSD) and Concavity of Distribution Function Property (CDFP) (e.g., \cite{laffont2009theory, szalay2009contracts, chade2014wealth}).
The former assumes that a costlier action 
leads to a better outcome distribution, in terms of first-order stochastic dominance.
The latter corresponds to having diminishing returns from the agent's costs, i.e., its efforts.

Under these two assumptions, we study the query complexity of learning an approximately optimal bounded contract in two models.
The \emph{action query model} (\Cref{sec:action-query}) allows the principal to sample from any action's outcome distribution.
This corresponds to the historical data of previous projects for which the agent's actions, and thus, the action-outcome pairs are observable.
The \emph{contract query model} (\Cref{sec:contract-query}) lets the principal specify a contract, and then provides an outcome sampled from the distribution of the agent's best-response action.
This is closely related to learning in the bandit model.
Indeed, a learning algorithm with polynomial query complexity in the contract query model implies an $O(T^c)$ no-regret learning algorithm in the bandit model.\footnote{Suppose that a learning algorithm $\calA$ uses $n(\varepsilon) = O(\varepsilon^{-c})$ queries to learn a contract that is optimal up to an $\varepsilon$ additive factor. Consider a corresponding bandit algorithm $\calB$ that divides the rounds into stages $i \ge 1$. In each stage $i$, it first runs algorithm $\calA$ for $n(2^{-i})$ rounds to learn a contract; 
then, it runs the learned contract for $n(2^{-i-1}) \cdot 2^i = O(2^{i(c+1)})$ rounds. This is a no-regret algorithm with regret bound $O(T^{\nicefrac{c}{c+1}})$.}

Our main results are two algorithms for learning an approximately optimal bounded contract using \emph{a polynomial number of queries} (\Cref{thm: action query,thm: contract query}) in these two models.
The query complexity results build on several technical ingredients, which we outline below.

\subsubsection*{Summary of Techniques}

The first ingredient is the representation of the outcome distributions by their complementary cumulative distribution functions (CDF).
It provides a more intuitive interpretation of the assumptions since FOSD and CDFP are equivalent to the monotonicity and concavity of the complementary CDFs.
Monotone and concave functions are structurally simpler than arbitrary functions and thus
easier to learn.
Indeed, we will exploit these structures to learn an approximation of the complementary CDFs and thus an approximation of the outcome distributions.

Our approach of learning an approximately optimal contract by learning an approximation of the underlying outcome distributions is in spirit related to the approach of learning optimal auctions through learning the bidders' value distributions~\cite{chen2023strong, guo2019settling, guo2021generalizing, jin2023learning}.
By contrast, existing learning algorithms~\cite{ho2016adaptive, cohen2022learning, zhu2023the} for contracts adopted the statistical learning theory approach and relied on discretizing the contract space into as few representative contracts as possible. 

Our second ingredient is a \emph{robustification procedure} for converting the optimal bounded contract of an instance into a robustified version that works well on nearby instances, i.e., those that only differ in the outcome distributions within a small total variation distance or Kolmogorov distance.
This ingredient is sufficient for getting polynomial query complexity in the action query model.

The contract query model is more challenging, because we cannot specify an action to sample an outcome from, and do not even know the set of possibly infinitely many actions.
Our next ingredient is a class of \emph{threshold contracts} that we use for extracting information about the complementary CDFs in the contract query model.
A threshold contract is defined by a threshold outcome $\omega$ and a payment amount $r$:
the agent gets paid by $r$ if the project outcome is at least as good as $\omega$ in terms of the principal's value.
We observe that querying a threshold contract naturally corresponds to querying a subgradient oracle of the complementary CDF's inverse for the threshold outcome.
We further design an algorithm for learning monotone and concave functions using polynomially many subgradient queries.
These queries provide enough information for approximating the complementary CDFs, except for those actions with costs close to $0$.

Our final ingredient, given the approximate complementary CDFs from the previous step, is a query strategy that either finds a good enough contract, or learns extra information about the instance because the agent's best-response action is a new action whose outcome distribution is bounded away from all previously observed ones.
In other words, we are in a win-win situation.
We further prove that there could be at most polynomially many such new actions under FOSD and CDFP.
Putting the ingredients together shows the polynomial query complexity bound.

\subsubsection*{Approximate Optimality of Bounded Contracts}

Our second set of results addresses the second question.
Concretely, we show that bounded contracts cannot provide a constant multiplicative approximation compared to the optimal (general) contract (\Cref{thm:hardness-multiplicative}), nor can they offer an additive approximation for an additive factor smaller than $\nicefrac{1}{4}$ (\Cref{thm:hardness-additive}).
These inapproximability results hold even if the instance satisfies the assumptions we will make for our learning results.
Further, we complement these negative findings by observing that the bounded contracts do provide a mixed approximation, which loses both a multiplicative factor and an additive factor when compared to the optimal (general) contract.
The mixed approximation is asymptotically tight, in terms of the tradeoff between the multiplicative and additive factors.

\subsection{Further Related Work}

\paragraph{Learning Optimal Contract.}
\cite{duetting2023optimal} considered linear contracts, a well-studied subclass of bounded contracts, in a noise-free model where the principal can noiselessly obtain its expected utility under a linear contract through a single query.
They gave an algorithm with $O(\log\log T)$ regret in this model.
On the empirical side, \cite{wang2023deep} considered automated contract design using deep learning. 
Finally, concurrent to this paper, \cite{bacchiocchi2023learning} presents an algorithm with $O(T^{\nicefrac{4}{5}})$ regret when the agent's action space is small.
Their results and ours are incomparable.
They do not assume FOSD and CDFP but only allow a constant number of actions.
In contrast, our results hold even if there are infinitely many actions, under the FOSD and CDFP assumptions. \cite{dutting2021the} also discussed the query complexity under the action query model, when the outcome space is exponentially large. Their algorithm outputs a near-optimal contract that is approximately incentive-compatible. Compared to theirs, our algorithms return a near-optimal contract with a strict incentive compatibility guarantee.

\paragraph{Simple and Approximately Optimal Contract.}
\cite{dutting2019simple} considered linear contracts, and showed multiplicative approximation ratios that depend on the number of actions, the max-to-min ratio of principal values, or the max-to-min ratio of the agent's costs.
\cite{castiglioni2021bayesian, guruganesh2021contracts} studied contract design when the agent's action and type are \emph{both} hidden.
They showed that finding the optimal contract is computationally intractable, and analyzed the approximation guarantees of linear contracts. 
Since linear contracts are a subclass of bounded contracts, the former's approximation guarantees also hold for the latter.
Indeed, we will adopt the mixed approximation guarantee by \cite{castiglioni2021bayesian}.
However, the negative results for linear contracts cannot directly apply to bounded contracts.
In \cite{alon2021contracts, alon2023bayesian}, the agent has single-dimensional type space. \cite{alon2021contracts} show that the optimal contract can be computed with a constant number of actions. \cite{alon2023bayesian} show the near optimality of linear contracts when there is sufficient uncertainty in the setting. In our work, we focus on instances with fixed types. We consider general bounded contracts and allow the agent's action space to be infinite.

\paragraph{Combinatorial Contracts}
\cite{babaioff2012combinatorial} studies the structural properties of optimal contracts in multi-agent settings. \cite{dutting2023multi} studies the tractability of computing optimal contracts when the principal's reward function on the set of participating agents satisfies some combinatorial properties. \cite{dutting2022combinatorial, dutting2023combinatorial} study single-principal single-agent cases, where the agent's action space is exponentially large. In our work, under the FOSD and CDFP assumptions, we allow the agent even to have a continuum of actions. \cite{vuong2023supermodular, ezra2023on} adding to the above two settings with further complexity results. \cite{dutting2021the} also considers the single-agent setting, where the outcome space is exponentially large in the input size.

 \section{Preliminaries} \label{sec:prelim}

\paragraph{Notations.} We use $\log$ and $\ln$ for logarithms in base $2$ and natural base, respectively.
Given any two distributions $D$ and $D'$ and any $0 \le \lambda \le 1$, we write $\lambda D + (1-\lambda) D'$ for the distribution that samples from $D$ with probability $\lambda$, and from $D'$ otherwise.

\subsection{Basic Model}

Consider the principal-agent problem with one principal and one agent.
The principal has a project which has $m$ possible outcomes $\Omega = \{0, 1, \dots, m-1\}$, and has a nonnegative value $v_\omega$ for each outcome $\omega \in \Omega$.
We assume that the values are normalized to be between $0$ and $1$ and sorted in ascending order, i.e., $0 \le v_0 \le v_1 \le \dots \le v_{m-1} \le 1$. 

The agent chooses from a set of actions $A$, which we will also refer to as its action space.
Each action $a \in A$ costs $c_a \ge 0$ for the agent and leads to a distribution of outcomes $D_a$.
Let $f(\omega | a)$ denote the probability of having outcome $\omega$ when the agent takes action $a$.
We assume that there is a null action, denoted as $0$, for which $c_0 = 0$ and $D_0$ is a point mass at outcome $0$, i.e., $f(0|0) = 1$.
Finally, let $c_{\max} = \sup_{a \in A} c_a$ denote the maximum cost.
 
We focus on the \emph{hidden-action} model (a.k.a., the \emph{moral hazard} model) of contract design. 
In this model, the principal cannot observe the agent's action, but only the induced outcome.
The principal first commits to a contract, represented by a payment vector $p = (p_\omega)_{\omega \in \Omega}$ that specifies the payment to the agent for each outcome.
We consider the \emph{limited liability} model, which means that $p_\omega \ge 0$ for every $\omega \in \Omega$.

Given a contract $p$, the agent chooses an action $a^*(p)$ that maximizes its expected utility.
The agent's utility equals the expected payment it receives minus the cost of its action.
Hence:
\[
    a^*(p) \in \argmax_{a \in A} ~ \left\{ ~ \sum_{\omega \in \Omega} f(\omega|a) \cdot p_\omega - c_a ~ \right\}
    ~.
\]
We assume that the action space and outcome distributions are such that the above best-response action is well-defined.
We say that contract $p$ incentivizes action $a^*(p)$.
Following the standard practice in the literature, we assume that the agent breaks ties \emph{in favor of the principal} when there is more than one action that maximizes its utility.

Correspondingly, the agent's utility under a contract $p$, given the above best-response, is:
\[
    \agentUtility(p) \coloneqq \max_{a \in A} ~ \left\{ ~ \sum_{\omega \in \Omega} f(\omega|a) \cdot p_\omega - c_a ~ \right\} = \sum_{\omega \in \Omega} f(\omega|a^*(p)) \cdot p_\omega - c_{a^*(p)}
    ~.
\]
In other words, the model guarantees \emph{Incentive Compatibility} (IC) for the agent.
Further, by our assumption of the existence of a null action $0$, IC implies \emph{Individual Rationality} (IR) for the agent, i.e., $\agentUtility(p) \ge 0$, because taking the null action leads to a non-negative utility for the agent. 

The principal's expected utility for a contract $p$ equals its expected value for the outcome minus the expected payment to the agent, when the agent takes the best-response action, i.e.:
\[
    \principalUtility(p) \coloneqq \sum_{\omega \in \Omega} f( \omega|a^*(p) ) \left( v_\omega - p_\omega \right)
    ~.
\]
The principal wants to design a contract to maximize its utility.
Let $\OPT$ denote the principal's optimal utility, i.e.:
\[
    \OPT \coloneqq \max_p ~ \principalUtility(p)
    ~.
\]

Recall that the principal's values for the outcomes are normalized to be between $0$ and $1$.
As a result, we may assume without loss of generality that the costs of the agent's actions are upper bounded by $1$, i.e., $c_{\max} \le 1$.
This is because the principal would not choose a contract $p$ that incentivizes an action $a^*(p)$ with cost $c_{a^*(p)} > 1$.
More concretely, the principal would need to make an expected payment strictly greater than $1$ in order to incentivize the action, due to the agent's IR property.
However, doing so would yield a negative utility for the principal since its expected value for the outcome is at most $1$.
This is dominated by choosing a null contract with $p_\omega = 0$ for all outcomes $\omega \in \Omega$.

\subsection{Bounded Contracts}
We say a contract $p = (p_\omega)_{\omega \in \Omega}$ is \emph{$H$-bounded} for some $H \ge 1$ if the payments satisfy $p_\omega \le H$ for all outcomes $\omega \in \Omega$.
When $H = 1$, we simply say that the contract $p$ is \emph{bounded}.
Although we normalize the values, and thus, may assume without loss of generality that the costs of the agents' actions are bounded, we stress that the optimal contract may not be bounded in general, which we shall see later in \Cref{sec:bounded-approximation}. 
This is precisely the starting point of this paper.
Let $\OPT_H$ denote the principal's maximum utility achievable by $H$-bounded contracts, i.e.:
\[
    \OPT_H \coloneqq \max_{p : ~ \forall \omega \in \Omega, ~ p_\omega \le H} ~ \principalUtility(p)
    ~.
\]

We say that $H$-bounded contracts give an  $\alpha$-multiplicative-approximation if:
\[
    \OPT_H \ge \frac{1}{\alpha} \cdot \OPT
    ~.
\]

Similarly, we say that $H$-bounded contracts give a $\beta$-additive-approximation if:
\[
    \OPT_H \ge \OPT - \beta
    ~.
\]

We remark that the optimal ($H$-bounded) contract can be computed in polynomial time if the set of actions is finite, and we have full information on the contract design instance, including the set of outcomes, the principal's values for the outcomes, the agent's set of actions, their costs, and the outcome distribution of each action.
We sketch the algorithm in \Cref{app:computation} for completeness.

\subsection{Query Complexity Models}

Next, we consider models in which the outcome distributions associated with the actions and even the set of actions are unknown.
The principal can only obtain information about them through historical data and/or repeated trial-and-error, modeled as access to the following two oracles.

\paragraph{Action Query Model.}
This model assumes that the set of actions $A$ is \emph{finite} and their costs are known, but the outcome distributions $D_a$'s are not.
The principal may acquire information about the outcome distributions through an \emph{action oracle}.  
In each query, the principal specifies an action $a \in A$, and then receives an I.I.D. sample $\omega \sim D_a$ from the oracle.

\paragraph{Contract Query Model.}
In this model, the set of  (potentially infinitely many) actions $A$ and the outcome distributions $D_a$'s are both unknown.
The principal may acquire information through a \emph{contract oracle}. 
More precisely, the principal specifies a contract $p$, and receives an I.I.D. sample $\omega \sim D_{a^*(p)}$, where recall that $a^*(p)$ is the agent's best-response action that maximizes the agent's utility under contract $p$.
Note that, however, the action $a^*(p)$ is unobservable to the principal.

\paragraph{Learning Optimal (Bounded) Contracts.}
A learning algorithm in one of the above models interacts with the corresponding oracle for a finite number of rounds.
In each round, it adaptively chooses a query based on the observed outcomes from all previous queries.
At the end, the learning algorithm outputs an ($H$-bounded) contract $p$.

The \emph{query complexity} for learning the optimal ($H$-bounded) contract in one of the above models is the minimum number of queries to the corresponding oracle that is sufficient for ensuring that with probability at least $1-\delta$, the principal's expected utility $\principalUtility(p)$ for the output contract $p$ is at least $\OPT - \varepsilon$ (respectively, $\OPT_H - \varepsilon$) for $0 < \varepsilon < 1/2$.
Here, the probability is over the randomness of both the learning algorithm and the oracle.

\subsection{Assumptions and Alternative Representation of an Instance}

Finally, we introduce two assumptions from the existing literature on (algorithmic) contract design. 
In \Cref{sec:learnability}, we prove polynomial query complexity bounds under these assumptions, which are necessary because of an exponential lower bound by \cite{zhu2023the}.
Subsequently, in \Cref{sec:bounded-approximation}, our inapproximability results hold even when the instance satisfies these assumptions.
The (mixed) approximation guarantee, by contrast, does not rely on them.

\begin{definition}
    An outcome distribution $D$ \emph{first-order stochastically dominates} another outcome distribution $D'$, denoted as $D \succeq D'$, if for any outcome $\omega \in \Omega$, the probability of sampling an outcome from $D$ that is at least as good as $\omega$ is greater than or equal to the counterpart for $D'$, i.e.:
    \[
        \Pr_{\omega' \sim D} \left[ \omega' \ge \omega \right] \ge \Pr_{\omega' \sim D'} \left[ \omega' \ge \omega \right]
        ~.
    \]
\end{definition}

\begin{assumption}[First-Order Stochastic Dominance (FOSD)]
    \label{assumption:fosd}
    For any actions $a, a' \in A$ such that $c_a \ge c_{a'}$, we have $D_a \succeq D_{a'}$.
\end{assumption}

In other words, a costlier action would lead to a better outcome distribution under FOSD.
We remark that FOSD is weaker than another assumption called \emph{Monotone Likelihood Ratio Property (MLRP)}\footnote{MLRP assumes that for any actions $a, a' \in A$ with $c_a \le c_{a'}$, the likelihood ratio $\frac{f(\omega | a')}{f(\omega | a)}$ is non-decreasing in $\omega$.}
in the literature.

Further, two actions $a$ and $a'$ with the same costs $c_a = c_{a'}$ have the same outcome distributions under FOSD, because the inequality in the definition of FOSD holds in both directions for $a$ and $a'$ and thus must hold with equality.
Hence, we may assume without loss of generality that the actions are a set of distinct costs, and will use $f(\omega|a)$ and $f(\omega|c_a)$ interchangeably under FOSD.

\begin{assumption}[Concavity of Distribution Function Property (CDFP)]
    For any actions $a, a', a'' \in A$ such that $c_{a''} = \lambda c_a + (1-\lambda) c_{a'}$ for some $0 \le \lambda \le 1$, we have $D_{a''} \succeq \lambda D_a + (1-\lambda) D_{a'}$.
\end{assumption}

That is, interpolating between two actions would not give a better outcome distribution if there was already an action with the same (expected) cost for the agent.
It also corresponds to having diminishing returns per unit of the agent's cost.

\paragraph{Representing Outcome Distributions by the Complementary CDFs.}
For any cost $c_a$ of an action $a \in A$ and any outcome $\omega \in \Omega$, consider the complementary CDF of the outcome distribution $D_a$, i.e., $c_a \mapsto \sum_{\omega' \ge \omega} f(\omega' | c_a)$.
We further consider its concave closure $F(\omega | c)$, which is a function defined on $0 \le c \le c_{\max}$.
Finally, we extend it to have value $F(\omega|c) = F(\omega|c_{\max})$ for any $c_{\max} \le c \le 1$, which is equivalent to adding a set of dummy actions dominated by the existing actions, and thus, does not affect the instance.
We can interpret the two assumptions as natural properties of the complementary CDFs.
CDFP corresponds to having $c_a \mapsto \sum_{\omega' \ge \omega} f(\omega' | c_a)$ on the concave closure for any outcome $\omega \in \Omega$ and any action $a \in A$.
FOSD corresponds to the monotonicity of the complementary CDFs and thus the concave closures $F(\omega | c)$.

This representation of the outcome distributions will be useful for our results, especially those regarding the query complexity.
Under our assumptions, for any cost $c \in [0, 1]$, there is a distribution of actions with expected cost $c$ and the outcome distribution defined by $F(\cdot|c)$.
Intuitively, it is sufficient to learn these monotone and concave functions $F(1|c) \ge  F(2|c) \ge \dots \ge F(m-1|c)$, which map the agent's cost to the probability of getting an outcome that is at least as good as some threshold outcome $1 \le \omega \le m-1$.
See \Cref{fig:CCDF-representation} for an illustration.

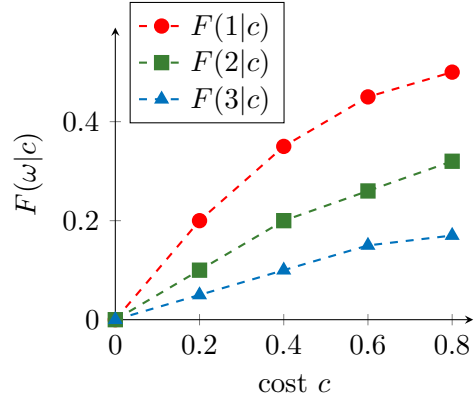
\begin{figure}
    \centering
    \begin{subfigure}{.4\textwidth}
        \centering
        \renewcommand{\arraystretch}{1.25}
        \vspace{24pt}
        \begin{tabular}{cccc}
            \toprule
            Cost & $f(1 | c)$ & $f(2 | c)$ & $f(3 | c)$ \\
            \midrule 
            $0.2$ & $0.1$ & $0.05$ & $0.05$ \\
            $0.4$ & $0.15$ & $0.1$ & $0.1$ \\
            $0.6$ & $0.19$ & $0.11$ & $0.15$ \\
            $0.8$ & $0.18$ & $0.15$ & $0.17$ \\
            \bottomrule
        \end{tabular}
        \vspace{24pt}
        \caption{Table representation}
    \end{subfigure}
    \begin{subfigure}{.48\textwidth}
        \centering
        \begin{tikzpicture}
        \begin{axis}[
            width=.8\textwidth,
            axis lines = left,
            xmax=0.85,
            ymax=0.59,
            xlabel = {cost $c$},
            ylabel = {$F(\omega | c)$},
            legend style={
                at={(0.04,0.88)},
                anchor=west
            }
        ]
        \addplot[red, thick, dashed, mark=*, mark options={solid}, 
        every mark/.append style={scale=1.25}
        ] coordinates {
            (0,0) (0.2,0.2) (0.4,0.35) (0.6,0.45) (0.8,0.5)
        };
        \addlegendentry{$F(1|c)$}
        \addplot[OliveGreen, thick, dashed, mark=square*, mark options={solid}, 
        every mark/.append style={scale=1.25}
        ] coordinates {
            (0,0) (0.2,0.1) (0.4,0.2) (0.6,0.26) (0.8,0.32)
        };
        \addlegendentry{$F(2|c)$}
        \addplot[RoyalBlue, thick, dashed, mark=triangle*, mark options={solid}, 
        every mark/.append style={scale=1.25}
        ] coordinates {
            (0,0) (0.2,0.05) (0.4,0.1) (0.6,0.15) (0.8,0.17)
        };
        \addlegendentry{$F(3|c)$}
        \end{axis}
        \end{tikzpicture}
        \caption{Complementary CDF representation}
    \end{subfigure}
    \caption{Representation of a problem instance under FOSD and CDFP assumptions in the table form and in the complementary CDF form.}
    \label{fig:CCDF-representation}
\end{figure}
 
 \section{Learnability of Bounded Contracts} \label{sec:learnability}

\paragraph{Further Notations.}
Throughout this section, we will constantly need to consider two instances $I$ and $\tilde{I}$, where usually $I$ is the original instance, and $\tilde{I}$ is an empirical instance learned from the queries. 
They share the same set of outcomes and principal's values.
However, they may have different action spaces, action costs, and outcome distributions.
In instance $I$, the action space is $A$, the action costs are $c_a$, and the outcome distributions are $D_a$.
In instance $\tilde{I}$, the action space is $\tilde{A}$.
We will use $\tilde{a}$ to denote an action from $\tilde{A}$, and denote its action cost and outcome distribution as $\tilde{c}_{\tilde{a}}$ and $\tilde{D}_{\tilde{a}}$.
Further, we write $\OPT_H(I)$ and $\OPT_H(\tilde{I})$ for the optimal principal's expected utilities achievable by $H$-bounded contracts under instance $I$ and $\tilde{I}$. 

\subsection{Metrics for Distributions}
\label{sec:distribution metrics}
We first introduce two metrics for distributions and some standard facts.
Consider any distributions $D$ and $\tilde{D}$ over the same support $\Omega = \{0, 1, \dots, m-1\}$.
Let $f$ and $\tilde{f}$ denote their probability mass functions.
Let $F(\omega) = \sum_{\omega' \ge \omega} f(\omega')$ and similarly $\tilde{F}(\omega) = \sum_{\omega' \ge \omega} \tilde{f}(\omega')$.

\begin{definition}[Total Variation Distance]
    The \emph{total variation distance} between $D$ and $\tilde{D}$ is
    \[
        \TV (D, \tilde{D}) \coloneqq \frac{1}{2} \sum_{\omega \in \Omega} \left| \, f(\omega) - \tilde{f}(\omega) \, \right|
        ~.
    \]
\end{definition}

\begin{definition}[Kolmogorov Distance]
    The \emph{Kolmogorov distance} between $D$ and $\tilde{D}$ is
    \[
        \Kol (D, \tilde{D}) \coloneqq \max_{\omega \in \Omega} ~ \left| \, F(\omega) - \tilde{F}(\omega) \, \right|
        ~.
    \]
\end{definition}

The following convergence bound of the empirical distribution to the true distribution in terms of the total variation distance is folklore.

\begin{lemma}[e.g., \cite{canonne2020short}]
    \label{lem:tv-convergence}
    Suppose that $\tilde{D}$ is the uniform distribution over $N = \Omega \left( \frac{m+\log\nicefrac{1}{\delta}}{\varepsilon^2} \right)$ IID samples from $D$.
    With probability at least $1-\delta$, we have
    \[
        \TV(D, \tilde{D}) \le \varepsilon
        ~.
    \]
\end{lemma}

It is also well-known that the difference between the expected values of a bounded function on two distributions is upper bounded by their total variation distance.

\begin{lemma}
    \label{lem:tv-expectation}
    For any function $h : \Omega \to [-H, H]$, we have
    \[
        \left| \, \E_{\omega \sim D} \, h(\omega) - \E_{\omega \sim \tilde{D}} \, h(\omega) \,\right| \le 2H \cdot \TV(D, \tilde{D})
        ~.
    \]
\end{lemma}

We also have a similar bound for the Kolmogorov distance, losing an additional factor $m$.

\begin{lemma}
    \label{lem:kol-expectation}
    For any function $h : \Omega \to [-H, H]$, we have
    \[
        \left| \, \E_{\omega \sim D} \, h(\omega) - \E_{\omega \sim \tilde{D}} \, h(\omega) \,\right| \le (2m-1) H \cdot \Kol(D, \tilde{D})
        ~.
    \]
\end{lemma}

We defer the proof of \Cref{lem:kol-expectation} to \Cref{app:distribution metrics}.

\subsection{Action Query Model} 

\label{sec:action-query}

We start by considering the query complexity of learning optimal contracts in the action query model. 
Recall that this model assumes a known and finite action space, and known action costs.
We can query any action $a \in A$ and receive a sample outcome $\omega$ from the outcome distribution $D_a$ of action $a$.
We denote the size of the action space by $n$.
\Cref{fig:action-query} presents the algorithm.

\begin{figure}
\begin{boxedtext}
\begin{enumerate}[leftmargin=*,itemsep=1mm]
    \item \emph{Sampling:~} Query each action $a \in A$ (for a sufficiently large constant $C$)
    \[
        C \cdot \frac{H^2 \left(m + \log\nicefrac{n}{\delta}\right)}{\varepsilon^4}
    \]
    times.
    Let $\tilde{D}_a$ be the uniform distribution 
    over the outcomes.
    \item \emph{Computing the Empirically Optimal Contract:~} Compute the optimal $H$-bounded contract $\tilde{p}^*$ for an instance whose outcome distributions are $\tilde{D}_a$, $a \in A$.
    \item \emph{Robustification:~} Return the \emph{robustified contract} $p$ such that for any $\omega \in \Omega$, 
    \begin{equation}
        \label{eqn: robustified contract}
        p_\omega = \tilde{p}^*_\omega + \frac{\varepsilon}{2} \left(v_\omega - \tilde{p}^*_\omega\right)
        ~. 
    \end{equation}
\end{enumerate}
\end{boxedtext}
    \caption{Psuedo-code for the learning algorithm in the action query model.}
    \label{fig:action-query}
\end{figure}

The robustification step is inspired by a price rounding technique that was commonly attributed to Nisan (see, e.g., \cite{chawla2007algorithmic, cai2013simple}).
On top of the original contract's payment, the robustified contract further shares an $\nicefrac{\varepsilon}{2}$ fraction of the principal's utility to the agent, which intuitively incentivizes the agent to take an action with a higher principal's 
utility 
in the original contract. 
For $H \geq 1$, it is clear that the robustified contract is also $H$-bounded.

To analyze the algorithm, we first define a notion of approximate action in another instance.

\begin{definition}
    For any instances $I$ and $\tilde{I}$ with the same outcome space $\Omega$, we say that an action $\tilde{a} \in \tilde{A}$ in instance $\tilde{I}$ has an \emph{$\varepsilon$-approximation} in instance $I$,
    if there is a mixed action $\sigma$, i.e., a distribution over $A$, such that
    \begin{enumerate}
        \item $\left| \, \tilde{c}_{\tilde{a}} - c_\sigma \, \right| \le \nicefrac{\varepsilon^2}{16}$ where $c_\sigma = \E_{a \sim \sigma}\, c_a$ is the expected action cost of sampling an action $a \sim \sigma$;
        \item either $\TV(\tilde{D}_{\tilde{a}}, D_\sigma) \le \nicefrac{\varepsilon^2}{32H}$ or $\Kol(\tilde{D}_{\tilde{a}}, D_\sigma) \le \nicefrac{\varepsilon^2}{32mH}$, where $D_\sigma$ is the outcome distribution obtained by first sampling an action $a \sim \sigma$, and then sampling an outcome from $D_a$.
    \end{enumerate}
\end{definition}

Next, we give some sufficient conditions under which the robustified contract works well on a different instance.
The lemma is more general than what is needed in this subsection because we will reuse this lemma in the contract query model later.

\begin{lemma}[Robustification]
    \label{lem:robustification}
    Consider two instances $I$ and $\tilde{I}$.
    Further consider the optimal $H$-bounded contract $\tilde{p}^*$ for instance $\tilde{I}$ and the robustified contract $p$ in \Cref{eqn: robustified contract}.
    Suppose for some $0 < \varepsilon < \nicefrac{1}{2}$, the action $\tilde{a}^* \in \tilde{A}$ incentivized by $\tilde{p}^*$ in instance $\tilde{I}$ has an $\varepsilon$-approximation in $I$.
    Then, either the principal's expected utility for contract $p$ in instance $I$ is at least
    \[
        \OPT_H(\tilde{I}) - \varepsilon
        ~,
    \]
    or the agent's best-response action $a^*(p)$ in instance $I$ does \emph{not} have an $\varepsilon$-approximation in $\tilde{I}$.
\end{lemma}

\begin{proof}
    Suppose that contract $p$ incentivizes action $a = a^*(p)$ in instance $I$, and contract $\tilde{p}^*$ incentivizes action $\tilde{a}^*$ in instance $\tilde{I}$.
    If $a$ has no $\varepsilon$-approximation in $\tilde{I}$, the lemma holds.
    
    Otherwise, let $\tilde{\sigma}$ be the mixed action in instance $\tilde{I}$ that $\varepsilon$-approximates $a$, and let $\sigma$ be the mixed action in instance $I$ that $\varepsilon$-approximates $\tilde{a}^*$.
    
    By the agent's IC property for the two instances, we have
    \begin{align}
        \label{eqn:ic-of-p}
        \E_{\omega \sim D_a} \, p_\omega - c_a
        &
        \geq
        \E_{\omega \sim D_\sigma} \, p_\omega - c_\sigma
        ~, \\[1ex]
        \label{eqn:ic-of-pstar}
        \E_{\omega \sim \tilde{D}_{\tilde{a}^*}} \, \tilde{p}^*_\omega - \tilde{c}_{\tilde{a}^*} 
        &
        \geq
        \E_{\omega \sim \tilde{D}_{\tilde{\sigma}}} \, \tilde{p}^*_\omega - \tilde{c}_{\tilde{\sigma}}
        ~.
    \end{align}

    Further, by the closeness of outcome distributions in the definition of $\varepsilon$-approximation, by the boundedness of contracts of $\tilde{p}^*$, and by \Cref{lem:tv-expectation,lem:kol-expectation}, we have
    \[
        \E_{\omega \sim \tilde{D}_{\tilde{a}^*}} \, \tilde{p}^*_\omega 
        \le
        \E_{\omega \sim D_\sigma} \, \tilde{p}^*_\omega + \frac{\varepsilon^2}{16}
        ~,\qquad
        \E_{\omega \sim \tilde{D}_{\tilde{\sigma}}} \, \tilde{p}^*_\omega
        \ge
        \E_{\omega \sim D_a} \, \tilde{p}^*_\omega - \frac{\varepsilon^2}{16}
        ~.
    \]
    and:
    \[
        \tilde{c}_{\tilde{a}^*} \ge c_\sigma - \frac{\varepsilon^2}{16}
        ~,\qquad
        \tilde{c}_{\tilde{\sigma}} \le c_a + \frac{\varepsilon^2}{16}
        ~.
    \]
    
    Combining with \Cref{eqn:ic-of-pstar}, we get
    \[
        \E_{\omega \sim D_\sigma} \, \tilde{p}^*_\omega - c_\sigma
        \geq
        \E_{\omega \sim D_a} \, \tilde{p}^*_\omega - c_a - \frac{\varepsilon^2}{4}
        ~.
    \]
    
    Further combining this with \Cref{eqn:ic-of-p} implies that
    \begin{equation}
        \label{eqn:robustification-last-step}
        \E_{\omega \sim D_a} \, (p_\omega - \tilde{p}^*_\omega) 
        \geq
        \E_{\omega \sim D_\sigma} \, (p_\omega - \tilde{p}^*_\omega) - \frac{\varepsilon^2}{4}
        ~.
    \end{equation}

    By \Cref{eqn: robustified contract}, we achieve that
    \[
        p_w - \tilde{p}^*_\omega = \frac{\varepsilon}{2} \left( v_\omega - \tilde{p}^*_\omega \right) = \frac{\varepsilon}{2-\varepsilon} \left( v_\omega - p_\omega \right)
        ~.
    \]

    Putting it back to \Cref{eqn:robustification-last-step}, we get that
    \begin{equation}
    \label{eqn: a and tilde-a under original instance}
        \E_{\omega \sim D_a} \, (v_\omega - p_\omega) 
        \geq
        \left(1-\frac{\varepsilon}{2}\right) \left( \E_{\omega \sim D_\sigma} \, (v_\omega - \tilde{p}^*_\omega) - \frac{\varepsilon}{2} \right)
        ~.
    \end{equation}

    Since $v_\omega \in [0,1]$, and $\tilde{p}^*_\omega \in [0,H]$, we have $v_\omega - \tilde{p}^*_\omega \in [-H,1] \subseteq [-H, H]$.
    By \Cref{lem:tv-expectation,lem:kol-expectation}, and that the mixed action $\sigma$ is an $\varepsilon$-approximation for $\tilde{a}^*$, we have
    \[
        \E_{\omega \sim D_\sigma} \, (v_\omega - \tilde{p}^*_\omega) \geq \E_{\omega \sim \tilde{D}_{\tilde{a}^*}} \, (v_\omega - \tilde{p}^*_\omega) - \frac{\varepsilon^2}{16}
        ~.
    \]

    Combining with \Cref{eqn: a and tilde-a under original instance}, we get
    \[
        \E_{\omega \sim D_a} \, (v_\omega - p_\omega) 
        \geq
        \left(1-\frac{\varepsilon}{2}\right) \left( \E_{\omega \sim \tilde{D}_{\tilde{a}^*}} \, (v_\omega - \tilde{p}^*_\omega) - \frac{\varepsilon}{2} - \frac{\varepsilon^2}{16} \right)
        ~.
    \]
    
    On the one hand, the left-hand side is the principal's expected utility for contact $p$ in instance $I$.
    On the other hand, the expectation on the right-hand side equals $\OPT_H(\tilde{I})$.
    Therefore, the lemma now follows by $(1-\nicefrac{\varepsilon}{2})(\OPT_H(\tilde{I})-\nicefrac{\varepsilon}{2}-\nicefrac{\varepsilon^2}{16}) \ge \OPT_H(\tilde{I})-\varepsilon$.
\end{proof}

With \Cref{lem:tv-convergence}, we can compute a query complexity upper bound through the above algorithm. The detailed proof of \Cref{thm: action query} is in \Cref{app:action query}.

\begin{theorem}
    \label{thm: action query}
    There is a polynomial-time algorithm for learning the optimal $H$-bounded contract in the action query model with query complexity at most
    \[
        O \left( \frac{nH^2 \left(m + \log \nicefrac{n}{\delta}\right)}{\varepsilon^4} \right)
        ~.
    \]
\end{theorem}

Now, let $\eta$ be the smallest non-zero probability of the outcome distributions in instance $I$:
\[
    \eta \coloneqq \min ~ \left\{ ~ f(\omega | a)\mid \omega \in \Omega, a\in A, f(\omega | a)>0 ~ \right\} ~.
\]
For general contracts, we make a simple observation that if $H$ is larger than $\nicefrac{1}{\eta}$, then $H$-bounded contracts are already optimal. The proof of \Cref{thm:opt-equal-opth} is in \Cref{app:opt-equal-opth}. 

\begin{theorem}
\label{thm:opt-equal-opth}
    For any $H \ge \nicefrac{1}{\eta}$, we have $\OPT = \OPT_H$.
\end{theorem}

With the above two results, we have the following corollary for learning the optimal (general) contract using action queries.

\begin{corollary}
    There is a polynomial-time algorithm for learning the optimal contract in the action query model with query complexity at most
    \[
        O \left( \frac{n \left(m + \log \nicefrac{n}{\delta}\right)}{\eta^2 \varepsilon^4} \right)
        ~.
    \]
\end{corollary}

\subsection{Contract Query Model} \label{sec:contract-query}

We next consider the query complexity of learning optimal contracts in the contract query model. 
Recall that in this model, we do not even know the action space, which may contain an infinite number of actions.
We can query any contract $p$ and receive a sample outcome $\omega$ from the outcome distribution $D_{a^*(p)}$ of the agent's best-response action $a^*(p)$.
Our main result is the next theorem.

\begin{theorem}
    \label{thm: contract query}
    Assuming FOSD and CDFP, there is a polynomial-time algorithm for learning the optimal $H$-bounded contract in the contract query model with query complexity at most
    \[
        \tilde{O} \left( \frac{m^{11}H^{10}}{\varepsilon^{20}} \right)
        ~.
    \]    
\end{theorem}

Similarly, 
we also have the following corollary for learning the optimal (general) contract.

\begin{corollary}
    Assuming FOSD and CDFP, there is a polynomial-time algorithm for learning the optimal contract in the 
    contract
    query model with query complexity at most
    \[
        \tilde{O} \left( \frac{m^{11}}{\eta^{10} \varepsilon^{20}} \right)
        ~.
    \]        
\end{corollary}

\paragraph{Algorithm Outline.}
The algorithm consists of two parts.
The first part is an \emph{Initialization} step for constructing an initial empirical instance $\tilde{I}$.
Recall that we may assume without loss of generality that the actions have distinct costs under FOSD and CDFP.
Hence, we will refer to an action by its action cost in the following discussion.
Ideally, we would like to learn the outcome distribution of each action cost $c$, up to the error bounds allowed by our \emph{Robustification} procedure (\Cref{lem:robustification}), in terms of either the total variation distance or the Kolmogorov distance.
If we could do that, then we would be able to apply the same approach as in the action query model to obtain a polynomial query complexity upper bound.
We identify a family of contracts that we call \emph{threshold contracts}, which allow us to effectively acquire information about the complementary CDF functions $F(\omega|c)$, viewing them as functions over the cost space.
See \Cref{sec:initialization} for details.
However, the acquired information is insufficient for pinning down the outcome distributions for action costs close to zero.
This is why we need the second part.

The second part is an \emph{Iterative Refinement} step. 
As discussed above, the robustified version of the optimal contract for the empirical instance $\tilde{I}$ may not provide a good enough principal's expected utility, due to insufficient information about the actions with costs close to zero and their outcome distributions.
Fortunately, we can easily verify whether the principal's expected utility is sufficiently large, since we can simply query the robustified contract through a polynomial number of contract queries.
Further, if the principal's expected utility is too low, then we are in the second case of the conclusion of \Cref{lem:robustification}, i.e., the agent's best-response action is sufficiently different from all actions that we currently have in the empirical instance.
Intuitively, this is a new piece of information using which we can refine our empirical instance to bring it one step closer to the true instance.
Indeed, we will iteratively refine our empirical instance using this approach, until the robustified optimal contract for the empirical instance gives a sufficiently large principal's expected utility.
Further, we show that there could be at most a polynomial number of iterations, utilizing the FOSD and CDFP assumptions. 
See \Cref{sec:refinement} for details.

We present the pseudo-code in \Cref{fig:contract-query}, with $C$ being a sufficiently large constant.
The rest of the subsection will detail the design and analysis of each component. The proof of \Cref{thm: contract query} is a conclusion utilizing all the lemmas with detailed computation. 

\begin{figure}
\begin{boxedtext}
    \emph{Step I: Initialization} 
    \begin{enumerate}
        \item Learn piece-wise linear functions $\tilde{F}(\omega|c)$ on $c \in [0, 1]$ such that for any 
        $c \ge \nicefrac{\varepsilon^2}{144}$
        \[
            F(\omega|c) \le \tilde{F}(\omega|c) \le F(\omega|c) + \frac{\varepsilon^2}{288mH}
            ~,
        \] 
        where the number of contract queries used is at most
        \[
            C \cdot \frac{m^{11}H^{10} \log \nicefrac{mH}{\delta \varepsilon} \log \nicefrac{mH}{\varepsilon}}{\varepsilon^{20}}
            ~.
        \]
        \item Define the initial empirical instance $\tilde{I}$ such that
        \begin{itemize}
            \item The action space is $\tilde{A} = \{0\} \cup [\nicefrac{\varepsilon^2}{144}, 1]$;
            \item Action $\tilde{a} = 0$ has cost $0$ and outcome distribution defined by $\tilde{F}(\omega|\nicefrac{\varepsilon^2}{144})$;
            \item Non-zero action $\tilde{a} \in \tilde{A}$ has cost $\tilde{a}$ and outcome distribution defined by $\tilde{F}(\omega|\tilde{a})$.
        \end{itemize}
    \end{enumerate}    
    \emph{Step II: Iterative Refinement} 
    \begin{enumerate}
        \item Compute the optimal contract $\tilde{p}^*$ for instance $\tilde{I}$.
        \item Construct a robustified contract $p$ according to \Cref{eqn: robustified contract}.
        \item Query contract $p$ for 
        \[
            C \cdot \frac{m^3 H^2 \log\nicefrac{mH}{\delta\varepsilon}}{\varepsilon^4}
        \]
        times to estimate the outcome distribution up to a total variation distance of 
        $\nicefrac{\varepsilon^2}{576mH}$.
        \item If the principal's estimated utility is at least $\OPT_H(\tilde{I}) - \nicefrac{\varepsilon}{2}$, return contract $p$.
        \item Otherwise, add a new action to the empirical instance $\tilde{I}$ with cost $0$ and the estimated outcome distribution.
        Then, repeat the \emph{Iterative Refinement} procedure.
    \end{enumerate}
\end{boxedtext}
    \caption{Psuedo-code for the learning algorithm in the contract query model.}
    \label{fig:contract-query}
\end{figure}

\subsubsection{Design and Analysis of the Initialization Step}
\label{sec:initialization}

In the \emph{Initialization} step, our goal is to find good approximations for the monotone and concave functions $F(\omega|c)$ on $0 \le c \le 1$, where $\omega \in \Omega$.
This can be further divided into two sub-steps. 
The first sub-step considers a subgradient query oracle and shows how to use it to approximate monotone convex/concave functions. 
The second sub-step further implements a subgradient query oracle for the complementary CDF functions under the contract query model. 

\paragraph{Learning Convex/Concave Functions Using Subgradient Queries.}
We first formally define the subgradient query oracle. 

\begin{definition}[Subgradient Query Oracle]
A $\varepsilon$-subgradient query oracle of a 
non-decreasing
convex function $G$ takes a positive $p$ as input and outputs a point $\tilde{x}$ such that $p$ is a subgradient of $G$ 
at some $z$ satisfying $\tilde{x} - \varepsilon \le z \le \tilde{x} $. 
\end{definition}

Next, we show how to learn convex/concave functions with a subgradient query oracle as defined above. We begin by presenting the first result of learning a convex function. 
The convex function we aim to learn will be the complementary CDF's inverse function.

\begin{lemma}[Learning Convex Functions]
    \label{lmm: approximate convex}
    For any $\varepsilon > 0$, given an $\nicefrac{\varepsilon^2}{16}$-subgradient oracle to an unknown non-decreasing convex function $G: [0,s] \rightarrow [0,1]$, where $0 < s \leq 1$ is also unknown,
    we can learn an approximate function $\tilde{G}$ with $O ( \nicefrac{1}{\varepsilon} \log \nicefrac{1}{\varepsilon} )$ queries, such that 
    for any $x \in [0, s]$, $G(x) \ge \tilde{G}(x)$, and further either $\tilde{G}(x) \ge G(x) - \varepsilon$, 
    or $\tilde{G}$ has a subgradient of at least $\nicefrac{1}{\varepsilon}$ at $x + \nicefrac{\varepsilon^2}{16}$. 
\end{lemma}

Now, to learn a concave function, intuitively, we can first learn its inverse according to \Cref{lmm: approximate convex}, and then take the inverse of the approximation to obtain an approximation on the original concave function. The following lemma formalizes such intuition. 

\begin{lemma}[Learning Concave Functions]
\label{lmm: approximate concave}
    For any $\varepsilon > 0$, given an $\nicefrac{\varepsilon^4}{64}$-subgradient oracle to the inverse of an unknown non-decreasing concave function $F:[0,1] \rightarrow [0,1]$,
    we can learn an approximate function $\tilde{F}$ with $O(\nicefrac{1}{\varepsilon^2} \log \nicefrac{1}{\varepsilon})$ queries, such that $\tilde{F}(y) \geq F(y)$ on $[0, 1]$, and for any $y \in [\varepsilon, 1]$, $\tilde{F}(y) \leq F(y) + \varepsilon$. 
\end{lemma}

We defer the proofs of \Cref{lmm: approximate convex,lmm: approximate concave} to \Cref{sec:approximate-convex,app: approximate concave}, respectively.

\paragraph{Implementing a Subgradient Oracle via Threshold Contract Queries.}

We now proceed to 
realize a subgradient oracle under the contract query model, and subsequently learn a piece-wise linear approximation of concave functions $F(\omega|c)$ as instructed by \Cref{lmm: approximate concave}. We define a family of \emph{threshold contracts}. 

\begin{definition}[Threshold Contracts]
    An \emph{$(\omega, r)$-threshold contract} satisfies $p_{\omega'} = 0$ for any $\omega' < \omega$ and $p_{\omega'} = r$ for any $\omega' \geq \omega$. 
\end{definition}

We can relate threshold contract queries with a subgradient query oracle for the complementary CDFs' inverse functions, with the following important property. 

\begin{lemma}\label{lem:threshold-query-subgradient}
    Let $G_\omega$ be the inverse function of $F(\omega|c)$. Then, the optimal action cost choice $c^*$ responding to an $(\omega, r)$-threshold contract satisfies that $r$ is a subgradient of $G_\omega$ at $z = F(\omega | c^*)$. 
\end{lemma}

\begin{proof}
    Under the contract, the agent's utility with action cost $c$ is:
    \[
        \agentUtility (p,c) = r \cdot F(\omega | c) - c ~. 
    \]
    Thus, to maximize its utility, the optimal action cost 
    $c^*$ of the agent should satisfy that $\nicefrac{1}{r}$ is a supergradient of $F(\omega | c)$ at $c^*$. This is equivalent to 
    that $r$ is a subgradient of $G_\omega$ at $z = F(\omega | c^*)$. 
\end{proof}

Notice that we only observe a single outcome rather than the outcome distribution under the contract query model. Thus, we can boost \Cref{lem:threshold-query-subgradient} by conducting many identical threshold contract queries, which leads to a close approximation on $F(\omega | c_a)$, therefore realizing the subgradient oracle. As a result, the following lemma is a standard application of Hoeffding's inequality. 

\begin{lemma}\label{lem:constructing-subgradient-query}
    Suppose we make $O(\nicefrac{1}{\varepsilon^2} \log \nicefrac{1}{\delta})$ identical queries with the $(\omega,r)$-threshold contract, and let $\tilde{x}$ be the fraction of queries such that the resulting outcome is no worse than $\omega$ plus the semi-confidence bound $\nicefrac{\varepsilon}{2}$. Then with probability $1 - \delta$, there exists a point $z$ lying in $[\tilde{x} - \varepsilon, \tilde{x}]$ such that the subgradient of $G_\omega$ on $z$ is $r$. Here, $G_\omega$ is the inverse function of $F(\omega|c)$. 
\end{lemma}

Combining \Cref{lmm: approximate concave,lem:constructing-subgradient-query} together, we now conclude the design of the initialization step with the following result. We defer its proof to \Cref{app:conclude-initialization}.

\begin{lemma}
\label{lem:conclude-initialization}
    We can take 
    \[
        O\left(\frac{m^{11}H^{10} \log \nicefrac{mH}{\delta \varepsilon} \log \nicefrac{mH}{\varepsilon}}{\varepsilon^{20}}\right)
    \]
    contract queries to learn piece-wise linear functions $\tilde{F}(\omega|c)$ on $c \in [0, 1]$ such that with probability $1 - \delta$, for any $c \ge \nicefrac{\varepsilon^2}{144}$,
    \[
        F(\omega|c) \le \tilde{F}(\omega|c) \le F(\omega|c) + \frac{\varepsilon^2}{288mH}
        ~.
    \]
\end{lemma}

\subsubsection{Analysis of the Iterative Refinement Step}
\label{sec:refinement}

The analysis consists of two parts.
We will first state three invariants and show that the instances $I$ and $\tilde{I}$ satisfy the invariants at all times in the \emph{Iterative Refinement} step.
Then, we will bound the running time by proving that each iteration runs in polynomial time, and there can be at most polynomially many iterations.

\paragraph{Invariants.}

\begin{itemize}
    \item[(I1)] Every action $\tilde{a} \in \tilde{I}$ has an $\nicefrac{\varepsilon}{3}$-approximation in $I$.
    \item[(I2)] Every action $a \in I$ with $c_a \ge \nicefrac{\varepsilon^2}{144}$ has an $\nicefrac{\varepsilon}{3}$-approximation in $\tilde{I}$.
    \item[(I3)] For every action $a \in I$ with $0 \le c_a < \nicefrac{\varepsilon^2}{144}$, its outcome distribution $D_a$ is first-order stochastically dominated by $\tilde{D}_{\tilde{a}}$ for some $\tilde{a} \in \tilde{A}$ with cost $\tilde{c}_{\tilde{a}} = 0$.
\end{itemize}

\begin{lemma}\label{lmm: invariants}
    The instances $I$ and $\tilde{I}$ satisfy invariants (I1), (I2) and (I3) at all times in the \emph{Iterative Refinement} step.
\end{lemma}

\begin{proof}
    We will prove it by induction on the iterations.
    After the initialization step, invariants (I1) and (I2) are true by \Cref{lem:conclude-initialization} and that $\nicefrac{\varepsilon^2}{144}\geq \nicefrac{\varepsilon^2}{288mH}$.  
    Note that
    we include an action $\tilde{a} = 0$ with outcome distribution $\tilde{F}(\cdot | \nicefrac{\varepsilon^2}{144})$, which, by \Cref{lem:conclude-initialization}, first-order stochastically dominates $F(\cdot | \nicefrac{\varepsilon^2}{144})$, the outcome distribution of the action in $I$ with cost $\nicefrac{\varepsilon^2}{144}$. By FOSD, the outcome distribution of any action in $I$ with cost less than  $\nicefrac{\varepsilon^2}{144}$ is first-order stochastically dominated by $F(\cdot | \nicefrac{\varepsilon^2}{144})$. 
    Note that invariants (I2) and (I3) continue to hold in 
    the \emph{Iterative Refinement} step as $\tilde{I}$ includes more actions into its action space.
    It remains to prove invariant (I1) in the induction. 

    Suppose the invariant (I1) holds up to some iteration, and consider the next iteration.
    This means that the estimated principal utility under the robustified contract $p$, denoted as $\principalUtilityE(p)$, is strictly less than $\OPT_H(\tilde{I})-\nicefrac{\varepsilon}{2}$. 
    As the outcome distribution is estimated up to a total variation distance of $\nicefrac{\varepsilon^2}{576mH}$, by \Cref{lem:tv-expectation}, the true principal utility under $p$, $\principalUtility(p)$, satisfies
    \[
        \principalUtility(p) \leq \tilde{u}_{\rm{p}}(p) + \frac{\varepsilon^2}{288m} \leq \tilde{u}_{\rm{p}}(p) + \frac{\varepsilon}{6}
        ~.
    \]
    Therefore, the robustified contract $p$ yields strictly less than $\OPT_H(\tilde{I}) - \nicefrac{\varepsilon}{3}$ principal's utility in instance $I$.
    By \Cref{lem:robustification}, it must be the case that the agent's best-response action $a^*(p)$ in instance $I$ does not have an $\nicefrac{\varepsilon}{3}$-approximation in $\tilde{I}$.

    Further, by invariant (I2), we further conclude that the cost $c_{a^*(p)}$ is less than $\nicefrac{\varepsilon^2}{144}$.
    In other words, adding it to $\tilde{I}$ with cost $0$ satisfies the cost requirement of $\nicefrac{\varepsilon}{3}$-approximation in invariant (I1).
    Finally, by \Cref{lem:tv-convergence}, the number of queries to contract $p$ allows us to estimate the outcome distribution of $a^*(p)$ up to a total variation distance of
    $\nicefrac{\varepsilon^2}{576mH}\leq \nicefrac{\varepsilon^2}{288H}$. 
    Hence, we also satisfy the closeness of the outcome distributions in invariant (I1).
\end{proof}

We complete the first part of the analysis by proving that the invariants ensure that the empirical instance's optimal principal's utility is close to optimal, and thus, is a sufficiently good benchmark for achieving near optimality in the original instance.

\begin{lemma}
    \label{lem:contract-query-empirical-approximation}
    Under invariants (I1), (I2), and (I3), we have
    \[
		\OPT_H(\tilde{I}) \geq \OPT_H(I) - \frac{\varepsilon}{3}
        ~.
	\]
\end{lemma}

\begin{proof}
	Denote the optimal $H$-bounded contract under the original instance $I$ as $p^*$.
    Suppose that $p^*$ incentivizes action $a^*$ in instance $I$.
    If $c_{a^*} \ge \nicefrac{\varepsilon^2}{144}$, then \Cref{lem:robustification} (with the roles of $I$ and $\tilde{I}$ swapped) implies that there is an $H$-bounded contract for $\tilde{I}$ such that the principal's expected utility is at least
    \[ 
        \OPT_H(I) - \frac{\varepsilon}{3}
        ~.
    \]
    
    If $c_{a^*}<\nicefrac{\varepsilon^2}{144}$, then by invariant (I3) there is an action $\tilde{a}$ with cost $\tilde{c}_{\tilde{a}} = 0$ and an outcome distribution $\tilde{D}_{\tilde{a}}$ that first-order stochastically dominates $D_{a^*}$. 
    Consider a null contract in instance $\tilde{I}$, which has zero payment for any outcome. 
    By the agent's IC, it only chooses an action with cost 0. 
    Thus, since the agent breaks ties in favor of the principal, the principal's expected utility under this null contract is at least
    \[
        \E_{\omega \sim \tilde{D}_{\tilde{a}}}\, v_\omega \ge \E_{\omega \sim D_{a^*}}\, v_\omega \ge \OPT_H(I)
        ~.
    \]

\end{proof}

\paragraph{Computational Tractability.} To show that the \emph{Iteration Refinement} step stops in polynomial rounds, it suffices to show the efficient computability of optimal contracts for the empirical instances, and that the number of iterations is polynomial, which are concluded in the following lemmas.

\begin{lemma}
\label{lem: efficient computing of optimal contract}
    The optimal contract $\tilde{p}^*$ of the empirical instance $\tilde{I}$ can be computed in time polynomial in $m$, $H$, and $\nicefrac{1}{\varepsilon}$. 
\end{lemma}

While the empirical instance has a continuum of actions, the complementary CDFs are piece-wise linear functions. Therefore, it suffices to consider actions that are breakpoints of the complementary CDFs, as other actions must be dominated by some of them. The detailed proof is in \Cref{app: efficient computing of optimal contract}.

\begin{lemma}
\label{lem: maximum number of iteration}
    The \emph{iterative refinement} step consists of at most $O(\nicefrac{m^2H}{\varepsilon^2})$ iterations.
\end{lemma}

The lemma follows from the monotonicity of the complementary CDFs and that function values are bounded. 
The detailed proof is in \Cref{app: maximum number of iteration}.

\subsubsection{Discussion}

\subsubsection*{Comparison to Learning in Stackelberg Games}

While contract design can be seen as a special type of Stackelberg game, we stress that the existing results on learning in Stackelberg games are not directly applicable. 
Concretely, \cite{letchford2009learning} and \cite{peng2019learning} studied security games and bi-matrix games, where the follower has a \emph{finite} action set.
Further, the learning algorithm therein can query a principal's strategy and observe the agent's choice of action.
By contrast, under the contract query model, contract design (1) has 
possibly a \emph{continuum} of actions, and (2) the algorithm only observes the realized project outcome but not the agent's action. 

We also note that \cite{cacciamani2023online}, independent and concurrent to this work, modified the approach of \cite{peng2019learning} to address the second difficulty above, and obtained polynomial query complexity when the agent has only a constant number of actions.
Their algorithm learns the feasible region of contracts that incentivize each meta-action that corresponds to one or more actions with similar outcome distributions.
By contrast, our setting allows a continuum of unknown actions.
Our algorithm does not guarantee to find every action that can be incentivized by contracts, but only those relevant in defining an approximately optimal contract.

\subsubsection*{Dependence on FOSD and CDFP}
We use these two assumptions to show that threshold contracts can extract sufficient information about the complementary CDFs of the outcome distributions.
If FOSD does not hold, then threshold contracts may not provide information about the complementary CDF in the region where it is decreasing;
we also rely on FOSD to upper bound the number of iterations in the \emph{Iterative Refinement} stage (\Cref{lem: maximum number of iteration}).
Similarly, if CDFP fails, then threshold contracts may not offer information in the region where the complementary CDF disagrees with its concave closure.
We leave as an interesting open question whether other classes of contracts can extract information on the outcome distributions under weaker assumptions.

\subsubsection{Proof of \texorpdfstring{\Cref{thm: contract query}}{Theorem 3.8}}

\paragraph{Query Complexity.}
The number of queries used in the \emph{Initialization} step is the bottleneck and matches the stated query complexity bound.
The number of queries used in the \emph{Iterative Refinement} step is at most
\[
    \underbrace{\tilde{O} \left( \frac{m^3 H^2}{\varepsilon^4} \right)}_{\substack{\vphantom{\substack{\mbox{1}\\\mbox{2}}} \mbox{\footnotesize number of queries}\\\mbox{\footnotesize per iteration}}}
    \cdot 
    \underbrace{O \left( \frac{m^2 H}{\varepsilon^2} \right)}_{\substack{\vphantom{\substack{\mbox{1}\\\mbox{2}}} \mbox{\footnotesize number of iterations}\\\mbox{\footnotesize by \Cref{lem: maximum number of iteration}}}}
    = ~ \tilde{O} \left( \frac{m^5 H^3}{\varepsilon^6} \right)
    ~.
\]

\paragraph{Computational Complexity.}
By \Cref{lem: efficient computing of optimal contract}, we can compute the optimal contract for the empirical instance in polynomial time.
By \Cref{lem: maximum number of iteration}, the \emph{Iterative Refinement} step runs for at most a polynomial number of iterations.
Finally, the other steps run in polynomial time by definition.
Therefore, the learning algorithm runs in polynomial time.

\paragraph{Approximation Guarantee.}
Specifically, the algorithm will return a contract $p$ under which the principal's estimated utility $\tilde{u}_{\rm{p}}(p)$ is at least $\OPT_H(\tilde{I})-\nicefrac{\varepsilon}{2}$.
Further, since we have estimated the outcome distribution up to total variation distance $\nicefrac{\varepsilon^2}{576mH}$, by \Cref{lem:tv-expectation}, we estimated the principal's expected utility up to an additive error of $\nicefrac{\varepsilon^2}{576m} < \nicefrac{\varepsilon}{6}$. 
Therefore, the principal's expected utility under $p$ is at least
\[
    \OPT_H(\tilde{I})-\frac{\varepsilon}{2} - \frac{\varepsilon}{6} = \OPT_H(\tilde{I}) - \frac{2\varepsilon}{3}
    ~.
\]

Further by \Cref{lem:contract-query-empirical-approximation}, this is at least $\OPT_H(I)- \varepsilon$.

\subsubsection{Proof of \texorpdfstring{\Cref{lmm: approximate convex}}{Lemma 3.10}}
\label{sec:approximate-convex}

\begin{figure}
    \centering
    \begin{subfigure}{.44\textwidth}
        \centering
        \begin{tikzpicture}
        \begin{axis}[
            xmin=0, xmax=0.63, 
            ymin=0, 
            width=.95\textwidth,
            axis lines=left,
            ticks=none,
            line width=0.7pt,
            axis line style={-{Stealth[angle'=45]}},
            samples=500,
            clip=false
        ]
            \addplot [black,mark=none,line width=0.7pt,domain=0:0.54]{x^2};

            \addplot [black,mark=none,line width=0.7pt,domain=0.25:0.55,densely dashed]{2*0.4*(x-0.4)+(0.4)^2};
            \addplot [black,mark=none,line width=0.7pt,domain=0.3:0.6,densely dashed]{2*0.4*(x-0.4)+0.1};

            \addplot [black,mark=none,line width=0.7pt,domain=0.2:0.3]{2*0.2*(x-0.2)};
            \addplot [black,mark=none,line width=0.7pt,domain=0.3:0.4]{2*0.3*(x-0.3)+0.04};
            \addplot [black,mark=none,line width=0.7pt,domain=0.4:0.5]{2*0.4*(x-0.4)+0.1};
            \addplot [black,mark=none,line width=0.7pt,domain=0.5:0.6116]{2*0.5*(x-0.5)+0.18};

            \addplot [black,mark=none,line width=0.5pt,dashed]
                coordinates {
                    (0.2,0) (0.2,0.04)
                };
            \addplot [black,mark=none,line width=0.5pt,dashed]
                coordinates {
                    (0.3,0) (0.3,0.09)
                };
            \addplot [black,mark=none,line width=0.5pt,dashed]
                coordinates {
                    (0.4,0) (0.4,0.16)
                };
            \addplot [black,mark=none,line width=0.5pt,dashed]
                coordinates {
                    (0.5,0) (0.5,0.25)
                };
   
            \fill [black] (0.2, 0.04) circle (2pt); 
            \fill [black] (0.3, 0.09) circle (2pt);
            \fill [black] (0.4, 0.16) circle (2pt);
            \fill [black] (0.5, 0.25) circle (2pt);

            \fill [black] (0.2, 0) circle (2pt); 
            \fill [black] (0.3, 0.04) circle (2pt);
            \fill [black] (0.4, 0.1) circle (2pt);
            \fill [black] (0.5, 0.18) circle (2pt);

            \node at (0.42,0.25) {$G$};
            \node at (0.56,0.157) {$\hat{G}$};
            
        \end{axis}
        \end{tikzpicture}
        \caption{Relationship between $G$ and $\hat{G}$. }
    \end{subfigure}
    \begin{subfigure}{.44\textwidth}
        \centering
        \begin{tikzpicture}
        \begin{axis}[
            xmin=0.15, xmax=0.7, 
            ymin=0, 
            width=.95\textwidth,
            axis lines=left,
            ticks=none,
            line width=0.7pt,
            axis line style={-{Stealth[angle'=45]}},
            samples=500,
            clip=false
        ]

            \addplot [black,mark=none,line width=0.7pt,domain=0.2:0.3]{2*0.2*(x-0.2)};
            \addplot [black,mark=none,line width=0.7pt,domain=0.3:0.4]{2*0.3*(x-0.3)+0.04};
            \addplot [black,mark=none,line width=0.7pt,domain=0.4:0.5]{2*0.4*(x-0.4)+0.1};
            \addplot [black,mark=none,line width=0.7pt,domain=0.5:0.6116]{2*0.5*(x-0.5)+0.18};

            \addplot [black,mark=none,line width=0.7pt,domain=0.28:0.36]{2*0.2*(x-0.28)};
            \addplot [black,mark=none,line width=0.7pt,domain=0.36:0.47]{2*0.3*(x-0.36)+0.032};
            \addplot [black,mark=none,line width=0.7pt,domain=0.47:0.55]{2*0.4*(x-0.47)+0.098};
            \addplot [black,mark=none,line width=0.7pt,domain=0.55:0.6796]{2*0.5*(x-0.55)+0.162};

            \addplot [black,mark=none,line width=0.7pt,domain=0.3:0.6,densely dashed]{2*0.4*(x-0.4)+0.1};
            \addplot [black,mark=none,line width=0.7pt,domain=0.37:0.65,densely dashed]{2*0.4*(x-0.47)+0.098};
            
            \fill [black] (0.2, 0) circle (2pt); 
            \fill [black] (0.3, 0.04) circle (2pt);
            \fill [black] (0.4, 0.1) circle (2pt);
            \fill [black] (0.5, 0.18) circle (2pt);

            \fill [black] (0.28, 0) circle (2pt); 
            \fill [black] (0.36, 0.032) circle (2pt);
            \fill [black] (0.47, 0.098) circle (2pt);
            \fill [black] (0.55, 0.162) circle (2pt);

            \addplot [black,mark=none,line width=0.5pt,dashed]
                coordinates {
                    (0.2,0.2916) (0.2,0) 
                };
            \addplot [black,mark=none,line width=0.5pt,dashed]
                coordinates {
                    (0.28,0.2916) (0.28,0) 
                };
            \addplot [black,mark=none,line width=0.5pt,dashed]
                coordinates {
                    (0.3,0.2916) (0.3,0.04) 
                };
            \addplot [black,mark=none,line width=0.5pt,dashed]
                coordinates {
                    (0.36,0.2916) (0.36,0.032) 
                };
            \addplot [black,mark=none,line width=0.5pt,dashed]
                coordinates {
                    (0.4,0.2916) (0.4,0.1) 
                };
            \addplot [black,mark=none,line width=0.5pt,dashed]
                coordinates {
                    (0.47,0.2916) (0.47,0.098) 
                };
            \addplot [black,mark=none,line width=0.5pt,dashed]
                coordinates {
                    (0.5,0.2916) (0.5,0.18) 
                };
            \addplot [black,mark=none,line width=0.5pt,dashed]
                coordinates {
                    (0.55,0.2916) (0.55,0.162) 
                };

            \draw [decorate,decoration=calligraphic brace,line width=0.8pt] (0.2,0.27) -- (0.28,0.27);
            \draw [decorate,decoration=calligraphic brace,line width=0.8pt] (0.3,0.27) -- (0.36,0.27);
            \draw [decorate,decoration=calligraphic brace,line width=0.8pt] (0.4,0.27) -- (0.47,0.27);
            \draw [decorate,decoration=calligraphic brace,line width=0.8pt] (0.5,0.27) -- (0.55,0.27);

            \node at (0.525,0.25) {$\hat{G}$};
            \node at (0.6,0.14) {$\tilde{G}$};
            
        \end{axis}
        \end{tikzpicture}
        \caption{Relationship between $\hat{G}$ and $\tilde{G}$. }
    \end{subfigure}
    \caption{An illustration on the proof of \Cref{lmm: approximate convex}.}
    \label{fig:approximate-convex}
\end{figure}
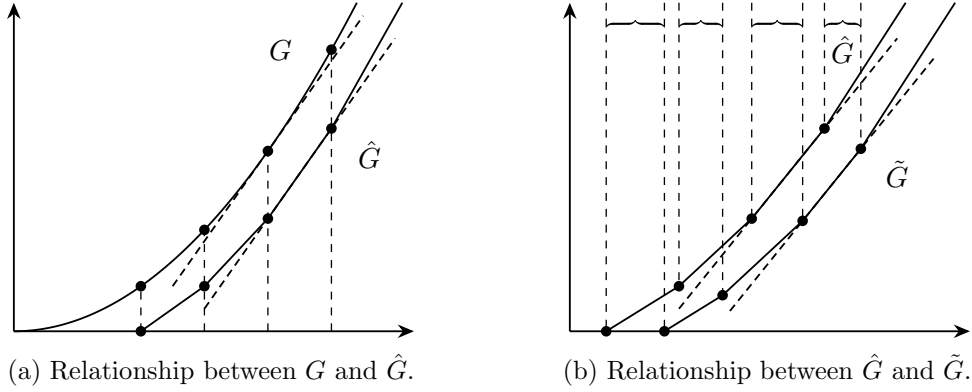
 
For a better understanding, we illustrate the main idea in the proof of \Cref{lmm: approximate convex} in \Cref{fig:approximate-convex}. 

Let $i_{\max} = \nicefrac{4}{\varepsilon} \ln \nicefrac{4}{\varepsilon}$. 
For any $-i_{\max} \leq i\leq  i_{\max}$, we query the subgradient $e^{\nicefrac{i\varepsilon}{4}}$, and let $\tilde{x}_i$ be the output of the subgradient oracle.
Further let $z_i$ be the point in $[\tilde{x}_i - \nicefrac{\varepsilon^2}{16}, \tilde{x}_i]$ such that $e^{\nicefrac{i\varepsilon}{4}}$ is a subgradient of $G$ at $z_i$, which is guaranteed to exist by the definition of the oracle.
However, notice that $z_i$ is unobservable to the algorithm and will only be used in the analysis. 
We artificially define $\tilde{x}_{-i_{\max}-1} = z_{-i_{\max}-1} = 0$ and $\tilde{x}_{i_{\max}+1} = z_{i_{\max}+1} =1$ for notational convenience. 

Next, we define an auxiliary approximate function $\hat{G} : [0, 1] \to [0, 1]$.
We first specify its value at $z_i$ recursively as follows, 
\begin{align*}
    \hat{G}(z_{-i_{\max}-1}) 
    = \hat{G}(z_{-i_{\max}}) = 0
    ~; \quad \hat{G}(z_{i+1})
    = \hat{G}(z_i) + e^{\nicefrac{i\varepsilon}{4}}\cdot (z_{i+1}-z_i) ~, \quad \forall -i_{\max} \leq i\leq i_{\max}
    ~.
\end{align*}
We extend the domain of $\hat{G}$ to $[0, 1]$ by linearly interpolating between $z_i$ and $z_{i+1}$ when $x \in (z_i, z_{i + 1})$.\footnote{For $x>z_{i_{\max}}$, We further extend it to $(z_{i_{\max}},s')$ using the line segment start at $(z_{i_{\max}},\hat{G}(z_{i_{\max}}))$ and has gradient  $e^{i_{\max}\varepsilon/4}$, where $s'$ is the point that the function value reaches $1$. This ensures that $\hat{G}$'s inverse has domain $[0,1]$ and matches that of the complementary CDFs. By \Cref{lmm:learn-convex-auxiliary}, $s'>s$, and thus $\hat{G}$ is well-defined on $[0,s]$. The same argument holds for $\tilde{G}$.  }

Recall that $G$ is non-decreasing and convex.
Hence, the subgradient of $G$ from $z_{-i_{\max}-1}$ to $z_{-i_{\max}}$ is at least $0$.
For any $-i_{\max} \le i \le i_{\max}$, the subgradient of $G$ from $z_i$ to $z_{i+1}$ is at least $e^{\nicefrac{i\varepsilon}{4}}$, a subgradient of $G$ at $z_i$.
Therefore, functions $G$ and $\hat{G}$ satisfy two properties summarized in the next lemma.
\begin{lemma}
    \label{lmm:learn-convex-auxiliary}
    The auxiliary approximate function $\hat{G}$ satisfies
    \begin{enumerate}
        \item For any $x \in [0, 1]$, $G(x) \ge \hat{G}(x)$;
        \item $G(x) - \hat{G}(x)$ is non-decreasing in $[0, 1]$.
    \end{enumerate}
\end{lemma}

This auxiliary function satisfies the approximation guarantee in the lemma.
However, we cannot compute it because we cannot observe $z_i$. 
Instead, we replace $z_i$ with the upper bound $\tilde{x}_i$ as guaranteed by the oracle.
The final approximate function $\tilde{G}$ is recursively defined at $\tilde{x}_i$ as
\begin{align*}
    \tilde{G}(\tilde{x}_{-i_{\max}-1})
    = \tilde{G}(\tilde{x}_{-i_{\max}}) = 0
    ~; \quad \tilde{G}(\tilde{x}_{i+1})
    = \tilde{G}(\tilde{x}_i) + e^{\nicefrac{i\varepsilon}{4}}\cdot (\tilde{x}_{i+1}-\tilde{x}_i) ~, \quad \forall -i_{\max} \leq i\leq i_{\max}
    ~.
\end{align*}
Furthermore, we extend the domain of $\tilde{G}$ to $[0, 1]$ by conducting linear interpolation. The following properties are natural.
\begin{lemma}
    \label{lmm:learn-convex-approximate}
    The approximate function $\tilde{G}$ satisfies
    \begin{enumerate}
        \item For any $x \in [0, 1]$, $\hat{G}(x) \ge \tilde{G}(x)$;
        \item $\hat{G}(x) - \tilde{G}(x)$ is non-decreasing in $[0, 1]$.
    \end{enumerate}
\end{lemma}

By combining \Cref{lmm:learn-convex-auxiliary,lmm:learn-convex-approximate}, we have $\tilde{G}(x) \leq G(x)$.
In the rest of the argument, it suffices to upper bound $G(x) - \tilde{G}(x)$ for $0 \le x \le z_{i_{\max}}$, because there is a subgradient of $\tilde{G}(x)$ at $z_{i_{\max}} + \nicefrac{\varepsilon^2}{16} \geq x_{i_{\max}}$ that is at least $\nicefrac{4}{\varepsilon}$ by our construction.

We will partition $G(x) - \tilde{G}(x)$ into two parts $G(x)-\hat{G}(x)$ and $\hat{G}(x)-\tilde{G}(x)$, and upper bound them separately. 
By the monotonicity of the differences (\Cref{lmm:learn-convex-auxiliary,lmm:learn-convex-approximate}), it is sufficient to bound $G(z_{i_{\max}})-\hat{G}(z_{i_{\max}})$ and $\hat{G}(z_{i_{\max}})-\tilde{G}(z_{i_{\max}})$. We defer the detailed computation of bounding these differences to \Cref{app: approximate convex}. Intuitively, the differences are sufficiently small as the gradients between $G(x)$ and $\hat{G}(x)$ are small, and the differences between $z_i$ and $\tilde{x}_i$ are small.

  \section{Approximate Optimality of Bounded Contracts} \label{sec:bounded-approximation}

This section studies whether bounded contracts could guarantee approximately optimal principal's utility compared to the best general contract. 
We start with two hardness results showing that bounded contracts could be far from optimal in terms of multiplicative or additive approximation, even 
under
FOSD and CDFP.
As a comparison, 
as
linear contracts are a special type of bounded contract, 
by \cite{castiglioni2021bayesian}, they 
give a mixed approximation 
without FOSD and CDFP.

\begin{theorem}[Hardness of Multiplicative Approximation]
    \label{thm:hardness-multiplicative}
    For any $H \ge 1$ 
    and any $\alpha > 1$, there is an instance satisfying FOSD and CDFP for which:
    \[
        \OPT_H < \frac{1}{\alpha} \cdot \OPT
        ~.
    \]
\end{theorem}

We present the contract instance and the principal's optimal utilities from general and bounded contracts below and defer the detailed proof to \Cref{app:hardness-multiplicative}.

\begin{proof}[Proof Sketch]
    We consider an instance parameterized by $\varepsilon > 0$, with $m = 3$ outcomes.
    The principal has value $v_0 = 0$ for outcome $0$ and $v_1 = v_2 = 1$ for outcomes $1$ and $2$.
    The agent has a continuum of actions with $0 \leq a \leq \ln \nicefrac{1}{\varepsilon}$, such that 
    $c_a = \varepsilon (e^a - 1 - a)$,    
    and the corresponding outcome distribution is defined by
    \begin{gather*}
        F(1|c_a) = \varepsilon e^a ~, \quad F(2|c_a) = \frac{\varepsilon}{2H} c_a = \frac{\varepsilon^2}{2H}\left(e^a - 1 - a\right)
        ~,
    \end{gather*}
    which are monotone and concave as functions of $c_a$.
    Hence, this instance satisfies FOSD and CDFP.
    Generally, the principal can incentivize the agent with a high payment for outcome $2$, and obtain an 
    expected utility of $\varepsilon (1 + \ln \nicefrac{1}{\varepsilon})$. 
    With $H$-bounded contracts,
    its expected utility is at most $3\varepsilon$.
\end{proof}

By the above proof, considering the relationships between $\varepsilon(1 + \ln \nicefrac{1}{\varepsilon})$ and $3\varepsilon$, we also have two corollaries that show hardness results for two forms of parameterized multiplicative approximation.

\begin{corollary}
\label{coro: lower bound of mixed approximation}
    There is an instance satisfying FOSD and CDFP such that
    \[
        \OPT_H = O \left(
        \frac{\OPT}{\log \nicefrac{1}{\OPT}} \right)
        ~.
    \]
\end{corollary}

\begin{corollary}
\label{coro: lower bound w.r.t. expected value}
     For any $0 < L < 1$, there is an instance satisfying FOSD and CDFP, and the principal's expected values of all actions are at least $L$, such that
     \[
        \OPT_H \le O \left( \frac{\OPT}{\ln \nicefrac{1}{L}} \right)
        ~. 
     \]
\end{corollary}

\begin{theorem}[Hardness of Additive Approximation]
    \label{thm:hardness-additive}
    For any $H > 0$ and any $\beta < 1/4$, there is an instance satisfying FOSD and CDFP for which:
    \[
        \OPT_H < \OPT - \beta
        ~.
    \]
\end{theorem}

We defer the detailed proof to \Cref{app:hardness-additive}.

\begin{proof}[Proof Sketch]
Consider an instance with $m = 3$ outcomes and $n = 3$ actions.
The principal has values $v_0 = 0$ and $v_1 = v_2 = 1$ for outcomes $1$ and $2$.
For a sufficiently small $\varepsilon > 0$, the agent's information is given in the following table.
The instance satisfies FOSD and CDFP. The principal's optimal expected utility is $\nicefrac{3}{4}$ with general contracts, but is at most $\nicefrac{1}{2}$ with $H$-bounded contracts.

\begin{center}
    \begin{tabular}{ccc}
        \toprule
        \text{\rm Cost} & $F(1|c)$ & $F(2|c)$ \\
        \midrule
        $c_1 = \varepsilon$ & $\nicefrac{1}{2}$ & $\nicefrac{4\varepsilon^2}{H}$ \\
        $c_2 = \nicefrac{1}{4}$ & $1$ & $\nicefrac{\varepsilon}{H}$ \\
        \bottomrule
    \end{tabular}
\end{center}

\end{proof}

We complement the hardness with a positive result, showing that bounded contracts satisfy a notion of mixed approximation.
Specifically, we observe that linear contracts, as defined below, are a special form of bounded contracts. 

\begin{definition}[Linear contract]\label{def:linear-contract}
	A $\rho$-linear contract $p$ 
 satisfies 
	$p_\omega = \rho \cdot v_\omega$ 
	for any outcome $\omega$. 
\end{definition}

Clearly, linear contracts are $1$-bounded. 
Denote the principal's optimal utility under linear contracts as $\lin$. The following result is given by \cite{castiglioni2021bayesian}. 

\begin{theorem}[\cite{castiglioni2021bayesian}]\label{thm: mixed approximation}
    For any $\varepsilon \in (0,1/2)$, we have
    \begin{equation*}
    	\OPT \leq 2 \left(\log \frac{1}{\varepsilon}\cdot \lin + \varepsilon \right)
    	~.
    \end{equation*}	
\end{theorem}

We prove it in \Cref{app:mixed-approximation} for completeness.
With $\varepsilon = \nicefrac{\OPT}{4}$ and $\nicefrac{L}{4}$, we get two 
parameterized multiplicative approximations, matching the hardness results in \Cref{coro: lower bound of mixed approximation,coro: lower bound w.r.t. expected value}.

\begin{corollary}
    \label{cor:mix-approx}
    For any instance, we have
	\[
        \OPT_{H=1} \ge \lin \ge \Omega \left( \frac{\OPT}{\log \nicefrac{1}{\OPT}} \right)
        ~.
    \]	
\end{corollary}

This means that bounded contracts would be a constant factor approximation if the optimal  $\OPT$ is lower bounded by some constant.
The impossibility of constant multiplicative approximation holds only when $\OPT$ is negligible.

\begin{corollary}
\label{coro: upper bound w.r.t. expected value}
    Suppose the principal's expected values of all actions are at least $L > 0$.
    We have
    \[
        \OPT_{H=1} \ge \lin \ge \Omega \left( \frac{\OPT}{\log \nicefrac{1}{L}} \right)
    	~.
    \]
\end{corollary}

We remark that \Cref{coro: upper bound w.r.t. expected value} is known:  
\cite{dutting2019simple}
 proves it using a different set of linear contracts.
They also showed that linear contracts are an $n$-multiplicative approximation compared to $\OPT$ where $n$ is the number of actions. 
However, this would be subsumed by \Cref{cor:mix-approx} as long as the optimal $\OPT$ is, at worst, exponentially small, i.e., $\OPT = \Omega(e^{-n})$.

\clearpage 

\appendix

\section{Computation of the Optimal Contracts with Full Information}
\label{app:computation}

This section assumes that the set of actions is finite, and we have full information on the contract design instance, including the set of outcomes, the principal's values for the outcomes, the agent's set of actions, their costs, and the outcome distribution of each action.
Then, we can compute the optimal general contract, and the optimal $H$-bounded contract for any given $H > 0$, in polynomial time.

For any action $a$, the set of contracts that incentivize action $a$ 
is a polytope defined by $p_w \ge 0$ for $\omega \in \Omega$ and the following set of linear constraints:
\[
    \forall a' \ne a ~: \quad \sum_{\omega \in \Omega} f(\omega|a) \cdot p_\omega - c_a \ge \sum_{\omega \in \Omega} f(\omega|a') \cdot p_\omega - c_{a'}
    ~.
\]
To further restrict our attention to $H$-bounded contracts, we only need to consider $m$ additional constraints $p_\omega \le H$ for $\omega \in \Omega$.
Therefore, we can compute the ($H$-bounded) contract that incentivizes action $a$ while maximizing the principal's utility, or equivalently, minimizing the expected payment to the agent, i.e., $\sum_{\omega \in \Omega} f(\omega|a) \cdot p_\omega$, by solving a linear program (LP).
Finally, we compare the principal's utilities for all actions $a \in A$ and choose the maximum one, and the corresponding optimal (bounded) contract.
 
\section{Missing Proofs from Section \ref{sec:learnability}}

\subsection{Proof of \Cref{lem:kol-expectation}}
\label{app:distribution metrics}

    Expanding the expectations, we have
    \begin{align*}
        \E_{\omega \sim D} \, h(\omega) - \E_{\omega \sim \tilde{D}} \, h(\omega)
        &
        = \sum_{\omega \in \Omega} \left( f(\omega) - \tilde{f}(\omega) \right) h(\omega) \\
        &
        = \left( F(0) - \tilde{F}(0) \right) h(0) + \sum_{\omega = 1}^{m-1} \left( F(\omega) - \tilde{F}(\omega) \right) \left( h(\omega) - h(\omega-1) \right) \\
        &
        \le \Kol(D, \tilde{D}) \cdot H + (m-1) \cdot \Kol(D, \tilde{D}) \cdot 2H
        ~.
    \end{align*}

    This is $(2m-1)H \cdot \Kol(D, \tilde{D})$, so the lemma follows.

\subsection{Proof of \Cref{thm: action query}}
\label{app:action query}

    Let $I$ and $\tilde{I}$ denote the original instance and the empirical instance with outcome distributions $\tilde{D}_a$, $a \in A$.
    By \Cref{lem:tv-convergence} together with the union bound, we get that
    \[
        \TV(D_a, \tilde{D}_a) \le \frac{\varepsilon^2}{128H}
    \]
    for all actions $a \in A$ with probability at least $1-\delta$.
    By \Cref{lem:robustification}, the principal's expected utility for contract $p$ in instance $I$ is at least
    \[
        \OPT_H(\tilde{I}) - \frac{\varepsilon}{2}
        ~.
    \]

    Apply \Cref{lem:robustification} again with the roles of $I$ and $\tilde{I}$ swapped. We also have a contract for which the principal's expected utility in instance $\tilde{I}$ is at least $\OPT_H(I)  - \nicefrac{\varepsilon}{2}$. This implies
    \[
        \OPT_H(\tilde{I}) \ge \OPT_H(I) - \frac{\varepsilon}{2}
        ~.
    \]
    
    Combining the above two inequalities proves the theorem.

\subsection{Proof of \Cref{thm:opt-equal-opth}}
\label{app:opt-equal-opth}

    Suppose for contradiction that the optimal contract $p$ must have a payment $p_\omega > \nicefrac{1}{\eta}$ for some outcome $\omega \in \Omega$. 
    Then, the agent's action incentivized by $p$ must have a non-zero probability of realizing outcome $\omega$.
    However, this means that the expected payment to the agent is strictly larger than $1$.
    Hence, the principal's expected utility is negative, which is impossible.

\subsection{Proof of \Cref{lmm: approximate concave}}
\label{app: approximate concave}
Let $G = F^{-1}$ be the inverse function of $F$, which is a convex function.
Further, let $\tilde{G}$ be the approximate function that we learn through \Cref{lmm: approximate convex}, replacing $\varepsilon$ therein with $\nicefrac{\varepsilon^2}{2}$.
Finally, let $\tilde{F}$ be the inverse of $\tilde{G}$.
Since $\tilde{G} \le G$, we have $\tilde{F} \geq F$. 
We will next show that for any $y \in [\varepsilon, 1]$, we have $\tilde{F}(y) \le F(y) + \varepsilon$. 
Equivalently, for any $y \in [\varepsilon, 1]$, let $x$ and $\tilde{x}$ be the points that 
\[
    G(x)=\tilde{G}(\tilde{x})=y
    ~.
\]
We need to bound $\tilde{x} - x$. 

We start by defining some key breakpoints in the argument, regarding when the subgradients of functions $G$ and $\tilde{G}$ are sufficiently large.
\begin{itemize}
    \item Let $\tilde{x}_{\max}$ satisfies that $\nicefrac{2}{\varepsilon^2}$ is a maximum subgradient of $\tilde{G}$ at $\tilde{x}_{\max}$.
    \item By the construction of $\tilde{G}$ in  \Cref{lmm: approximate convex}, there is a $z_{\max} \in [\tilde{x}_{\max} - \nicefrac{\varepsilon^4}{64}, \tilde{x}_{\max}]$ such that $\nicefrac{2}{\varepsilon^2}$ is a subgradient of $G$ at $z_{\max}$.
    \item Let $y_{\max} = \tilde{G}(z_{\max})$.
    \item Define $x_{\max}$ such that $G(x_{\max}) = y_{\max}$.
    \item Finally, let $\tilde{y}_{\max} = \tilde{G}(\tilde{x}_{\max}) \geq \tilde{G}(z_{\max}) = y_{\max}$. 
\end{itemize}

By \Cref{lmm: approximate convex}, we know that for any $\tilde{x} \leq z_{\max}$, 
\[
    \tilde{G}(\tilde{x}) \leq G(\tilde{x}) \leq \tilde{G}(\tilde{x}) + \frac{\varepsilon^2}{2}
    ~.
\]

We consider three possible cases based on the relationship among $\varepsilon$, $y_{\max}$, and $\tilde{y}_{\max}$. 
We first discuss the main case where $\bm{y_{\max} \ge \varepsilon}$. 
Two key scenarios in the discussion below are illustrated in \Cref{fig:learning-concave}. 

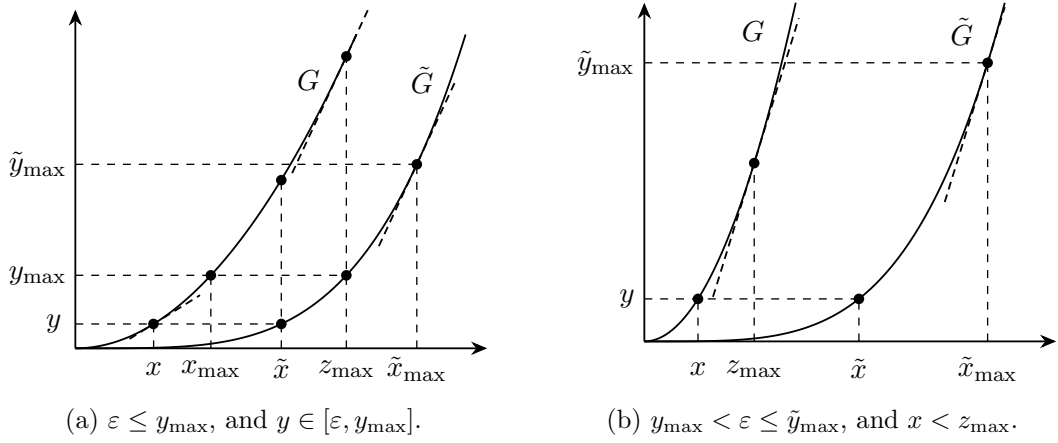
\begin{figure}
    \centering
    \begin{subfigure}{.45\textwidth}
        \centering
        \begin{tikzpicture}
        \begin{axis}[
            xmin=0, xmax=0.76, 
            ymin=0, 
            width=.95\textwidth,
            axis lines=left,
            ticks=none,
            line width=0.7pt,
            axis line style={-{Stealth[angle'=45]}},
            samples=500,
            clip=false
        ]
            \addplot [black,mark=none,line width=0.7pt,domain=0:0.5184]{x^2};
            \addplot [black,mark=none,line width=0.7pt,domain=0:0.72]{x^4}; 
            
            \addplot [black,mark=none,line width=0.7pt,domain=0.56:0.7,densely dashed]{4*(0.63)^3*(x-0.63)+(0.63)^4};
            \addplot [black,mark=none,line width=0.7pt,domain=0.4:0.54,densely dashed]{2*0.5*(x-0.5)+(0.5)^2};
            \addplot [black,mark=none,line width=0.7pt,domain=0.1:0.23,densely dashed]{2*0.144*(x-0.144)+(0.144)^2};
            
            \addplot [black,mark=none,line width=0.5pt,dashed]
                coordinates {
                    (0.63,0) (0.63,0.15752961)
                };
            \addplot [black,mark=none,line width=0.5pt,dashed]
                coordinates {
                    (0,0.15752961) (0.63,0.15752961)
                };
            \addplot [black,mark=none,line width=0.5pt,dashed]
                coordinates {
                    (0.5,0) (0.5,0.25)
                };
            \addplot [black,mark=none,line width=0.5pt,dashed]
                coordinates {
                    (0,0.0625) (0.5,0.0625)
                };
            \addplot [black,mark=none,line width=0.5pt,dashed]
                coordinates {
                    (0.25, 0) (0.25,0.0625) 
                };
            \addplot [black,mark=none,line width=0.5pt,dashed]
                coordinates {
                    (0.38, 0) (0.38,0.1444) 
                };
            \addplot [black,mark=none,line width=0.5pt,dashed]
                coordinates {
                    (0, 0.02085136) (0.38,0.02085136) 
                };
            \addplot [black,mark=none,line width=0.5pt,dashed]
                coordinates {
                    (0.1444, 0) (0.1444,0.02085136) 
                };

            \fill [black] (0.63, 0.15752961) circle (2pt); 
            \fill [black] (0.5, 0.25) circle (2pt);
            \fill [black] (0.5, 0.0625) circle (2pt);
            \fill [black] (0.25, 0.0625) circle (2pt);
            \fill [black] (0.38, 0.02085136) circle (2pt);
            \fill [black] (0.144, 0.02085136) circle (2pt);
            \fill [black] (0.38, 0.144) circle (2pt);

            \node at (0.63,-0.0175) {$\tilde{x}_{\max}$};
            \node at (-0.07,0.15752961) {$\tilde{y}_{\max}$};
            \node at (0.5,-0.0175) {$z_{\max}$};
            \node at (0.25,-0.0175) {$x_{\max}$};
            \node at (-0.07,0.0625) {$y_{\max}$};
            \node at (0.38,-0.0175) {$\tilde{x}$};
            \node at (0.1444,-0.0175) {$x$};
            \node at (-0.04,0.02085136) {$y$};

            \node at (0.43,0.23) {$G$};
            \node at (0.64,0.23) {$\tilde{G}$};
            
        \end{axis}
        \end{tikzpicture}
        \caption{$\varepsilon \leq y_{\max}$, and $y \in [\varepsilon, y_{\max}]$.}
    \end{subfigure}
    \begin{subfigure}{.45\textwidth}
        \centering
        \begin{tikzpicture}
        \begin{axis}[
            xmin=0, xmax=0.48, 
            ymin=0, 
            width=.95\textwidth,
            axis lines=left,
            ticks=none,
            line width=0.7pt,
            axis line style={-{Stealth[angle'=45]}},
            samples=500,
            clip=false
        ]
            \addplot [black,mark=none,line width=0.7pt,domain=0:0.1764]{x^2};
            \addplot [black,mark=none,line width=0.7pt,domain=0:0.42]{x^4}; 
            
            \addplot [black,mark=none,line width=0.7pt,domain=0.35:0.42,densely dashed]{4*(0.4)^3*(x-0.4)+(0.4)^4};
            \addplot [black,mark=none,line width=0.7pt,domain=0.08:0.18,densely dashed]{2*0.128*(x-0.128)+(0.128)^2};
            
            \addplot [black,mark=none,line width=0.5pt,dashed]
                coordinates {
                    (0.4,0) (0.4,0.0256)
                };
            \addplot [black,mark=none,line width=0.5pt,dashed]
                coordinates {
                    (0,0.0256) (0.4,0.0256)
                };
            \addplot [black,mark=none,line width=0.5pt,dashed]
                coordinates {
                    (0.128,0) (0.128,0.016384)
                };
            \addplot [black,mark=none,line width=0.5pt,dashed]
                coordinates {
                    (0.25,0) (0.25,0.00390625)
                };
            \addplot [black,mark=none,line width=0.5pt,dashed]
                coordinates {
                    (0,0.00390625) (0.25,0.00390625)
                };
            \addplot [black,mark=none,line width=0.5pt,dashed]
                coordinates {
                    (0.0625,0) (0.0625,0.00390625)
                };

            \fill [black] (0.4, 0.0256) circle (2pt); 
            \fill [black] (0.128,0.016384) circle (2pt); 
            \fill [black] (0.25,0.00390625) circle (2pt);
            \fill [black] (0.0625,0.00390625) circle (2pt);

            \node at (0.4,-0.0025) {$\tilde{x}_{\max}$};
            \node at (-0.045,0.0256) {$\tilde{y}_{\max}$};
            \node at (0.128,-0.0025) {$z_{\max}$};
            \node at (0.25,-0.0025) {$\tilde{x}$};
            \node at (0.0625,-0.0025) {$x$};
            \node at (-0.02,0.00390625) {$y$};
            
            \node at (0.127,0.0285) {$G$};
            \node at (0.37,0.0285) {$\tilde{G}$};
            
        \end{axis}
        \end{tikzpicture}
        \caption{$y_{\max} < \varepsilon \leq \tilde{y}_{\max}$, and $x < z_{\max}$.}
    \end{subfigure}
    \caption{Two key scenarios in the proof of \Cref{lmm: approximate concave}. (a) The first subcase of the \textbf{Main Case}. (b) \textbf{Boundary Case I}. }
    \label{fig:learning-concave}
\end{figure}

Specifically, we further consider three sub-cases depending on the value of $y$. 
\begin{enumerate}
\item $y \in [\varepsilon, y_{\max}]$.
    In this scenario, we have $\tilde{x} \leq z_{\max}$ because $\tilde{G}(\tilde{x}) = y \le y_{\max} = \tilde{G}(z_{\max})$.
    Hence, by \Cref{lmm: approximate convex}, we get that
    \[
        \tilde{G}(\tilde{x}) \leq G(\tilde{x}) \leq \tilde{G}(\tilde{x}) + \frac{\varepsilon^2}{2}
        ~.
    \]

    Let $g$ be a subgradient of $G$ at $x$. By the convexity of $G$:  
    \[
        g\cdot \left(\tilde{x}-x\right)\leq G(\tilde{x})-G(x) = G(\tilde{x})-\tilde{G}(\tilde{x}) \leq \frac{\varepsilon^2}{2}
        ~.
    \]
    We also have
    \[
        g \cdot \left(x-0\right) \geq G(x)-G(0) \geq \varepsilon
        ~,
    \]
    which implies
    \[
        g \ge \frac{\varepsilon}{x} \ge \varepsilon
        ~.
    \]
    Combining with the previous inequality, we derive that
    \[
        \tilde{x}-x \leq \frac{1}{g} \cdot \frac{\varepsilon^2}{2} \leq \frac{\varepsilon}{2}
        ~.
    \]
    \item $y \in (y_{\max}, \tilde{y}_{\max}]$.
        In this scenario, we have $\tilde{x} \leq \tilde{x}_{\max}$ because $\tilde{G}(\tilde{x}) = y \le \tilde{y}_{\max} = \tilde{G}(\tilde{x}_{\max})$.
        Meanwhile, we have $x > x_{\max}$ because $G(x) = y \ge y_{\max} = G(x_{\max})$.
        Therefore, we have
        \[
            \tilde{x} - x \leq \tilde{x}_{\max} - x_{\max} = \left(\tilde{x}_{\max} - z_{\max}\right) + \left(z_{\max} - x_{\max}\right)
            ~.
        \]

        By the relation between $\tilde{x}_{\max}$ and $z_{\max}$, they differ by at most $\nicefrac{\varepsilon^4}{64}$.
        Recall that $G(x_{\max}) = \tilde{G}(z_{\max}) = y_{\max}$.
        By the first sub-case, we have
        \[
            z_{\max} - x_{\max} \le \frac{\varepsilon}{2}
            ~.
        \]
        Combining everything together, we get that
        \[
            \tilde{x} - x \le \frac{\varepsilon^4}{64} + \frac{\varepsilon}{2} < \frac{3\varepsilon}{4}
            ~.
        \]
    \item $y \in (\tilde{y}_{\max}, 1]$.
        We have
        \[
            \tilde{x}-x = \left( \tilde{x} - \tilde{x}_{\max} \right) + \left(\tilde{x}_{\max} - x\right)
            ~.
        \]

        The first part is at most $\nicefrac{2}{\varepsilon^2}$ because $\nicefrac{2}{\varepsilon^2}$ is a subgradient of $\tilde{G}$ on $\tilde{x}_{\max}$, and thus, $\nicefrac{\varepsilon^2}{2}$ is a supergradient of $\tilde{F}$ on $\tilde{y}_{\max}$.
        The second part is at most $\nicefrac{3\varepsilon}{4}$ by the second sub-case, because $\tilde{G}(\tilde{x}) = \tilde{y}_{\max} < y = G(x)$.
        Putting together, we get that
        \[
            \tilde{x} - x \le \frac{\varepsilon^2}{2} + \frac{3\varepsilon}{4} \le \varepsilon
            ~.
        \]
\end{enumerate}

\subsubsection*{Boundary Case I: $\bm{y_{\max} < \varepsilon \leq \tilde{y}_{\max}}$.}

We have $x > x_{\max}$ and $\tilde{x} > z_{\max}$, because $G(x) = \tilde{G}(\tilde{x}) = y \ge \varepsilon > y_{\max} = G(x_{\max}) = \tilde{G}(z_{\max})$. 
If further $x < z_{\max}$, then by the convexity of function $G$, we have
\[
    \frac{z_{\max} - x}{z_{\max}} \leq \frac{G(z_{\max}) - G(x)}{G(z_{\max})} ~.
\]

Further note that $G(z_{\max}) - y_{\max} = G(z_{\max}) - \tilde{G}(z_{\max}) \leq \nicefrac{\varepsilon^2}{2}$ by the approximation guarantee of functions $G$ and $\tilde{G}$ at $z_{\max}$. 
Combining with $G(x) = y \geq \varepsilon$, the right-hand side is at most
\[
    \frac{y_{\max} + \nicefrac{\varepsilon^2}{2} - \varepsilon}{y_{\max} + \nicefrac{\varepsilon^2}{2}} < \frac{\nicefrac{\varepsilon^2}{2}}{\varepsilon + \nicefrac{\varepsilon^2}{2}} = \frac{\varepsilon}{2 + \varepsilon}
    ~.
\]

Therefore, we obtain that
\[
    x > z_{\max} \left( 1 - \frac{\varepsilon}{2 + \varepsilon} \right) \ge z_{\max} - \frac{\varepsilon}{2}
    ~.
\]

We now consider two sub-cases depending on the value of $y$. 
\begin{enumerate}
\item $y \in [\varepsilon, \tilde{y}_{\max}]$.
    Similar to the second sub-case of the \textbf{Main Case}, we have
\begin{align*}
    \tilde{x} - x \leq \tilde{x}_{\max} - x = \left(\tilde{x}_{\max} - z_{\max}\right) + (z_{\max} - x) \leq \frac{\varepsilon^4}{64} + \frac{\varepsilon}{2} < \frac{3}{4}\varepsilon ~. 
\end{align*}
\item $y \in (\tilde{y}_{\max}, 1]$.  
    Similar to the third sub-case of the \textbf{Main Case}, we have
    \begin{align*}
        \tilde{x}-x = \left(\tilde{x} - \tilde{x}_{\max}\right) + \left(\tilde{x}_{\max} - x\right) 
        \leq \frac{\varepsilon^2}{2} + \frac{3\varepsilon}{4}
        \leq \varepsilon
        ~.
    \end{align*}
\end{enumerate}

\subsubsection*{Boundary Case II: $\bm{\tilde{y}_{\max} < \varepsilon}$.}

We have $x > z_{\max}$  because $G(x) = y \ge \varepsilon > \tilde{y}_{\max} = \tilde{G}(\tilde{x}_{\max})\geq y_{\max} = \tilde{G}(z_{\max})$.
Thus, for an $y \in [\varepsilon, 1]$, we have
\begin{align*}
    \tilde{x}-x \leq \tilde{x} - z_{\max}
    \leq \left(\tilde{x} - \tilde{x}_{\max}\right) + \left(\tilde{x}_{\max} - z_{\max}\right)
    ~.
\end{align*}

The first part on the right at is most $\nicefrac{\varepsilon^2}{2}$ because $\nicefrac{2}{\varepsilon^2}$ is a subgradient of $\tilde{G}$ on $\tilde{x}_{\max}$, and thus, $\nicefrac{\varepsilon^2}{2}$ is a supergradient of $\tilde{F}$ on $\tilde{y}_{\max}$.
The second part is at most $\nicefrac{\varepsilon^2}{64}$ by the relation between $\tilde{x}_{\max}$ and $z_{\max}$.
Hence, we get that
\[
    \tilde{x}-x \leq  \frac{\varepsilon^2}{2} + \frac{\varepsilon^4}{64} 
    \leq \varepsilon
    ~.
\]

\subsection{Proof of \Cref{lem:conclude-initialization}}
\label{app:conclude-initialization}

Let $\delta' > 0$ be an auxiliary error probability, which we will use as the failure probability upper bound of each substep. 
Further, consider two auxiliary approximation parameters $\varepsilon_1, \varepsilon_2 > 0$.

By \Cref{lem:constructing-subgradient-query}, we can implement a single query to an $\varepsilon_1$-subgradient query oracle, with failure probability $\delta'$, using
\[
    O \left( \frac{\log \nicefrac{1}{\delta'}}{\varepsilon_1^2} \right)
\]
contract queries.
Further, let the auxiliary approximation parameters satisfy $\varepsilon_1 = \nicefrac{\varepsilon_2^4}{32}$.
By \Cref{lmm: approximate concave}, for each complementary CDF $F(\omega|c)$, we can learn function $F(\omega|c)$ up to an additive approximation error of $\varepsilon_2$ at any action cost $c \ge \varepsilon_2$, using 
\[
    O\left(\frac{\log \nicefrac{1}{\varepsilon_2}}{\varepsilon_2^2}\right)
\]
queries to the $\varepsilon_1$-subgradient query oracle.
By the relation between $\varepsilon_1$ and $\varepsilon_2$, since we need to learn $m$ such functions, this translates to
\[
    O\left(\frac{m \log \nicefrac{1}{\delta'} \log \nicefrac{1}{\varepsilon_2}}{\varepsilon_2^{10}}\right)
\]
contract queries, and a failure probability $\textrm{poly}(m, \nicefrac{1}{\varepsilon_2}) \cdot \delta'$ by union bound.

At last, notice that our robustification process (Lemma~\ref{lem:robustification}) requires that the Kolmogorov distance between the true and learned outcome distributions to be at most $O(\nicefrac{\varepsilon^2}{mH})$.
By letting $\varepsilon_2$ equal this value, we need
\[
    O\left(\frac{m^{11} H^{10} \log \nicefrac{1}{\delta'} \log \nicefrac{mH}{\varepsilon}}{\varepsilon^{20}}\right)
\]
contract queries, with the failure probability $\textrm{poly}(m, H, \nicefrac{1}{\varepsilon}) \cdot \delta'$.

The lemma now follows by letting $\delta' = \nicefrac{\delta}{\textrm{poly}(m, H, \nicefrac{1}{\varepsilon})}$.

\subsection{Proof of \Cref{lem: efficient computing of optimal contract}}

\label{app: efficient computing of optimal contract}

We mainly work on the increment representation of the values and the contracts, to simplify the notation. Specifically, let $v_{-1}=p_{-1}=0$. Given a contract $p$, we define for each $\omega \in \Omega$, 
\[
    u_\omega = v_\omega- v_{\omega - 1}~, ~~~~\text{and}~~~~
    q_\omega = p_{\omega}-p_{\omega-1}
    ~.
\]
Then we have 
\[
    \sum_{\omega \in \Omega} \tilde{f}(\omega | c)\cdot v_{\omega} = \sum_{\omega \in \Omega}\tilde{F}(\omega | c)\cdot u_{\omega}~, ~~~~\text{and}~~~~ \sum_{\omega \in \Omega} \tilde{f}(\omega | c)\cdot p_{\omega} = \sum_{\omega \in \Omega}\tilde{F}(\omega | c)\cdot q_\omega
    ~. 
\]

For $H$-bounded contracts, we have $u_\omega \in [0,1]$ and
\[
q_0 \in [0,H]~,~~\text{and}~~q_\omega \in [-H,H]~.
\]

Denote the total sum of the number of breakpoints of $\tilde{F}(\omega|c)$ as $k$. By construction, we know that $k = O(\nicefrac{m^3H^2\log\nicefrac{mH}{\varepsilon}}{\varepsilon^4})$.  We claim that every action $c$ that is not a breakpoint must be dominated by an action that is a breakpoint, which we will prove later. Therefore, label the breakpoints by $c_1\leq \cdots\leq c_k$. For each $c_i$, solve the following linear program (LP),
\begin{align*}
\min_{q \in [0,H]\times[-H,H]^{m-1}} \quad & \sum^{m-1}_{\omega=0}\tilde{F}(\omega|c_i)\cdot q_\omega, \\
\text{s.t.} \quad
& \sum^{m-1}_{\omega=0}\tilde{F}(\omega|c_i)\cdot q_\omega - c_i \geq \sum^{m-1}_{\omega=0}\tilde{F}(\omega|c_j)\cdot q_\omega - c_j, ~~\forall 1 \leq j \leq k.
\end{align*}
An optimal contract corresponds to an optimal solution that maximizes the principal utilities among the optimal solutions of all LPs. 

To prove the above claim, now consider an action $c=\alpha \cdot c_{i+1} + (1-\alpha)\cdot c_{i}$ with $0 < \alpha < 1$. Therefore $\tilde{F}(\omega |c) = \alpha \cdot \tilde{F}(\omega |c_{i+1}) + (1-\alpha)\cdot \tilde{F}(\omega |c_i)$, for any $\omega \in \Omega$. 
 For contract $p$, assume that $\sum^{m-1}_{\omega=0} \tilde{F}(\omega |c_{i+1})\cdot q_\omega - c_{i+1} \geq \sum^{m-1}_{\omega=0} \tilde{F}(\omega | c_{i})\cdot q_\omega - c_i$ without loss of generality. Then,
\begin{align*}
    &\mathrel{\phantom{=}} \left(\sum^{m-1}_{\omega=0}\tilde{F}(\omega|c_{i+1})\cdot q_\omega - c_{i+1}\right) - \left(\sum^{m-1}_{\omega=0}\tilde{F}(\omega|c)\cdot q_\omega - c\right)\\
&=~(1-\alpha) \left[\left(\sum^{m-1}_{\omega=0} \tilde{F}(\omega |c_{i+1})\cdot q_\omega - c_{i+1} \right) - \left(\sum^{m-1}_{\omega=0} \tilde{F}(\omega | c_{i})\cdot q_\omega - c_i\right)\right]\geq 0
    ~.
\end{align*}
As a result, $c$ is dominated by $c_{i+1}$. 

\subsection{Proof of \Cref{lem: maximum number of iteration}}
\label{app: maximum number of iteration}

    Suppose for contradiction that the algorithm runs for more than $\nicefrac{576m^2H}{\varepsilon^2}$ iterations.
    Then, more than $\nicefrac{576m^2H}{\varepsilon^2}$ actions $\tilde{a}_i$ have been added to the empirical instance $\tilde{A}$ with cost $0$.
    Further, for each action $\tilde{a}_i \in \tilde{A}$, let $a_i \in A$ be its $\nicefrac{\varepsilon}{3}$-approximation in the original action space, which is the agent's optimal action before we add action $\tilde{a}_i$ to $\tilde{A}$.

    We next argue that these actions $a_i$'s are bounded away from each other, and derive a contradiction by a packing argument.

    Consider the contract $p^i$ that we feed to the oracle before adding $\tilde{a}_i$ to $\tilde{A}$. Since we have estimated the outcome distribution up to total variation distance $\nicefrac{\varepsilon^2}{576mH}$, by \Cref{lem:tv-expectation}, we estimate the principal's expected utility up to an additive error of
    \[
        \frac{\varepsilon^2}{288m} < \frac{\varepsilon}{6}
        ~.
    \]

    Further, by that the estimated principal's expected utility of $p^i$ is at most $\OPT_H(\tilde{I}) - \nicefrac{\varepsilon}{2}$, we get that the actual expected utility of contract $p^i$ is less than
    \[
        \OPT_H(\tilde{I}) - \frac{\varepsilon}{2} + \frac{\varepsilon}{6} = \OPT_H(\tilde{I})  - \frac{\varepsilon}{3}
        ~.
    \]
    
    By the guarantee of the \emph{Robustification} procedure (\Cref{lem:robustification}), this implies that action $a_i$ does not have an $\nicefrac{\varepsilon}{3}$-approximation in the empirical instance $\tilde{I}$ before we add $\tilde{a}_i$.
    In particular, the outcome distribution $D_{a_i}$ and the outcome distribution $\tilde{D}_{\tilde{a}_j}$ of any action $\tilde{a}_j$ that is already in the empirical instance $\tilde{I}$ must have a Kolmogorov distance of at least
    \[
        \frac{\varepsilon^2}{288mH}
        ~.
    \]

    Further combining it with the fact that the Kolmogorov distance between $D_{a_j}$ and $\tilde{D}_{\tilde{a}_j}$ is at most $\nicefrac{\varepsilon^2}{576mH}$, we get that
    \[
        \Kol \left( D_{a_i}, D_{a_j} \right) \ge \frac{\varepsilon^2}{288mH} - \frac{\varepsilon^2}{576mH} = \frac{\varepsilon^2}{576mH}
        ~.
    \]

    Finally, suppose without loss of generality that the subscripts of $a_i$'s are in ascending order of the action costs.
    Then, combining the Kolmogorov distance with FOSD, for any $i$ we have
    \[
        \sum_{\omega \in \Omega} \left( F(\omega | a_{i+1}) - F(\omega | a_{i}) \right) \ge \max_{\omega \in \Omega} ~ \left( F(\omega | a_{i+1}) - F(\omega | a_{i}) \right) = \Kol\left(D_{a_i}, D_{a_{i+1}}\right) \ge \frac{\varepsilon^2}{576mH}
        ~.
    \]

    Since functions $F(\omega|a)$'s are bounded between $0$ and $1$ and there are $m$ outcomes $\omega \in \Omega$, there can be at most $\nicefrac{576m^2H}{\varepsilon^2}$ different actions subject to the above constraints.
    This leads to a contradiction and finishes the proof of the lemma.

\subsection{Detailed Computation of $G(z_{i_{\max}})-\hat{G}(z_{i_{\max}})$ and $\hat{G}(z_{i_{\max}}) - \tilde{G}(z_{i_{\max}})$ in \Cref{lmm: approximate convex}}
\label{app: approximate convex}
For the first part, we expand it as \begin{align*}
    \mathrel{\phantom{=}} G(z_{i_{\max}})-\hat{G}(z_{i_{\max}}) 
    = \left( G(z_{-i_{\max}})-\hat{G}(z_{-i_{\max}}) \right) + \sum_{i=-i_{\max}}^{i_{\max}-1} \left(\left(G(z_{i+1}) - G(z_i)\right) - \left( \hat{G}(z_{i+1}) - \hat{G}(z_{i})\right) \right)
    ~.
\end{align*}

Now, the subgradient of $G$ in $[0, z_{-i_{\max}}]$ is at most $e^{-\nicefrac{i_{\max} \varepsilon}{4}} = \nicefrac{\varepsilon}{4}$.
Hence, the first term on the right-hand side is at most
\[
    G(z_{-i_{\max}})-\hat{G}(z_{-i_{\max}}) = G(z_{-i_{\max}}) \le \frac{\varepsilon}{4} \cdot z_{-i_{\max}} \le \frac{\varepsilon}{4}
    ~.
\]

We next upper bound the summation.
The subgradient of $G$ in $[z_i, z_{i+1}]$ is at most $e^{\nicefrac{(i+1)\varepsilon}{4}}$, while the subgradient of $\hat{G}$ in this range equals $e^{\nicefrac{i\varepsilon}{4}}$.
Hence, they differ by at most an $e^{\nicefrac{\varepsilon}{4}}$ multiplicative factor.
Therefore
\begin{align*}
    \sum_{i=-i_{\max}}^{i_{\max}-1} \left(\left(G(z_{i+1}) - G(z_i)\right) - \left(\hat{G}(z_{i+1}) - \hat{G}(z_{i})\right) \right)
    & \le \left(1 - e^{-\nicefrac{\varepsilon}{4}}\right) \sum_{i=-i_{\max}}^{i_{\max}-1}\left(G(z_{i+1}) - G(z_i)\right) \\
    & \le \left(1 - e^{-\nicefrac{\varepsilon}{4}}\right) G(z_{i_{\max}}) \le 1 - e^{-\nicefrac{\varepsilon}{4}} \le \frac{\varepsilon}{4}
    ~.
\end{align*}

Putting together, we derive that
\[
    \hat{G}(z_{i_{\max}}) - G(z_{i_{\max}}) \le \frac{\varepsilon}{2}
    ~.
\]
We are left to consider the second part, i.e., $\hat{G}(z_{i_{\max}})-\tilde{G}(z_{i_{\max}})$.
By our construction of $\hat{G}$
\begin{align*}
    \hat{G}(z_{i_{\max}})
    &
    = \sum_{i=-i_{\max}}^{i_{\max}-1} \left(z_{i+1} - z_i\right) \cdot e^{\nicefrac{i\varepsilon}{4}} \\
    &
    = z_{i_{\max}} \cdot e^{\nicefrac{(i_{\max}-1)\varepsilon}{4}} - \sum_{i=-i_{\max}+1}^{i_{\max}-1} z_i \cdot \left( e^{\nicefrac{i\varepsilon}{4}} - e^{\nicefrac{(i-1)\varepsilon}{4}} \right) - z_{-i_{\max}}\cdot \frac{\varepsilon}{4}~.
\end{align*}

With a similar calculation, we have
\begin{align*}
    \tilde{G}(z_{i_{\max}})
    &
    = \left(z_{i_{\max}} - \min \left\{ \tilde{x}_{i_{\max} - 1}, z_{i_{\max}} \right\} \right)\cdot e^{\nicefrac{(i_{\max}-1)\varepsilon}{4}} \\
    & \qquad 
    + \sum_{i=-i_{\max}}^{i_{\max}-2} \left( \min \left\{ \tilde{x}_{i+1}, z_{i_{\max}} \right\} - \min \left\{ \tilde{x}_i, z_{i_{\max}} \right\} \right) \cdot e^{\nicefrac{i\varepsilon}{4}} \\
    &= z_{i_{\max}} \cdot e^{\nicefrac{(i_{\max}-1)\varepsilon}{4}} - \sum_{i=-i_{\max}+1}^{i_{\max}-1} \min \left\{ \tilde{x}_i, z_{i_{\max}} \right\} \cdot \left( e^{\nicefrac{i\varepsilon}{4}} - e^{\nicefrac{(i-1)\varepsilon}{4}} \right)\\
    & \qquad - \min \left \{\tilde{x}_{-i_{\max}}, z_{i_{\max}}\right\} \cdot \frac{\varepsilon}{4}~. 
\end{align*}

Hence, their difference is at most
\begin{align*}
    \hat{G}(z_{i_{\max}})-\tilde{G}(z_{i_{\max}}) &=\sum_{i=-i_{\max}+1}^{i_{\max}-1} \left( \min \left\{ \tilde{x}_i, z_{i_{\max}} \right\} - z_i \right) \left( e^{\nicefrac{i\varepsilon}{4}} - e^{\nicefrac{(i-1)\varepsilon}{4}} \right) \\
    & \qquad + \left (\min\left\{ \tilde{x}_{-i_{\max}}, z_{i_{\max}} \right\}-z_{-i_{\max}}\right)\cdot \frac{\varepsilon}{4} \\
    &\leq \frac{\varepsilon^2}{16} \cdot e^{\nicefrac{(i_{\max}-1) \varepsilon}{4}}  + \frac{\varepsilon}{4} 
    \le \frac{\varepsilon^2}{16} \cdot \frac{4}{\varepsilon} + \frac{\varepsilon}{4} 
    = \frac{\varepsilon}{2}
    ~.
\end{align*}
Thus, for any $0 \leq x \leq z_{i_{\max}}$, $G(x)-\tilde{G}(x) \leq G(z_{i_{\max}})-\tilde{G}(z_{i_{\max}}) \leq \varepsilon$. 
This proves the lemma. 
 \section{Missing Proofs from Section \ref{sec:bounded-approximation}}
\label{app:approximation}

\subsection{Proof of \texorpdfstring{\Cref{thm:hardness-multiplicative}}{Theorem 4.1}}
\label{app:hardness-multiplicative}

We first restate the instance below.
It is parameterized by $\varepsilon > 0$, with $m = 3$ outcomes.
The principal has value $v_0 = 0$ for outcome $0$ and $v_1 = v_2 = 1$ for outcomes $1$ and $2$.
The agent has a continuum of actions with $0 \leq a \leq \ln \nicefrac{1}{\varepsilon}$, such that 
\[
    c_a = \varepsilon (e^a - 1 - a) 
    ~,
\]
and the corresponding outcome distribution is defined by
\begin{gather*}
    F(1|c_a) = \varepsilon e^a ~, \quad F(2|c_a) = \frac{\varepsilon}{2H} c_a = \frac{\varepsilon^2}{2H}\left(e^a - 1 - a\right)
    ~.
\end{gather*}

We first bound the optimal principal's utility by general contracts.
Consider a contract $p$ with
\[
    p_0 = p_1 = 0 ~, \quad p_2 = \frac{2H}{\varepsilon}
    ~, 
\]
under which the agent gets zero utility from all actions. Breaking ties in favor of the principal, the agent takes action $a^*(p) = \ln \nicefrac{1}{\varepsilon}$.
The principal's utility is therefore  $\varepsilon (1 + \ln \nicefrac{1}{\varepsilon})$.

We next show that the principal's utility is at most $3 \varepsilon$ if the contract $p$ is $H$-bounded.
First, consider any contract $p$ with $p_1 > 1 - \varepsilon$.
The principal's utility is
\[
    f(1|a) \cdot \left( 1 - p_1 \right) + f(2|a) \cdot \left( 1 - p_2 \right)
    \le 1 \cdot \varepsilon + \frac{\varepsilon}{2H} \cdot 1 < \frac{3}{2}\varepsilon
    ~.
\]

Next, consider any contract with $p_1 \le 1-\varepsilon$.
The agent's best-response action $a^*(p)$ maximizes its utility, i.e., $f(1|a) \left(1-p_1\right) + f(2|a) \left(1-p_2\right) - c_a$, which equals
\[
    \varepsilon e^a \cdot p_1 + \frac{\varepsilon^2}{2H}\left(e^a - 1 - a\right)\cdot (p_2 - p_1) - \varepsilon(e^a - 1 - a)
    ~.
\]

Changing variables with $y = \nicefrac{\varepsilon(p_2 - p_1)}{2H}$, this is equivalent to maximize 
\[(1 - y)(1 + a) - (1 - p_1 - y)e^a ~.\]

Thus, we have 
\[
    a^*(p) = \ln ~ \min ~ \left\{ ~ \frac{1}{\epsilon} ~, \frac{1 - y}{1 - p_1 - y} ~ \right\}
    ~.
\]

For any action $a$, the principal's utility can be written as $F(1|a) \cdot \left(1-p_1\right) + F(2|a) \cdot \left(p_1-p_2\right)$.
The first term is $\varepsilon e^a \cdot (1 - p_1)$.
The second term is at most $\nicefrac{\varepsilon}{2}$ because $F(2|a) < \nicefrac{\varepsilon}{2H}$ and $p_1-p_2 \le H$.
Hence, for $a = a^*(p)$, the principal's utility is at most
\begin{align*}
    \varepsilon e^{a^*(p)} \cdot (1 - p_1) + \frac{\varepsilon}{2}
    \leq 
    \frac{\varepsilon (1 - y)(1 - p_1)}{1 - p_1 - y} + \frac{\varepsilon}{2} 
    \leq 
    2\varepsilon + \frac{\varepsilon}{2} < 3\varepsilon ~,
\end{align*}
where the last inequality follows by $p_1 \leq 1 - \varepsilon$ and $y \leq \nicefrac{\varepsilon}{2}$, which holds because the payments are bounded between $0$ and $H$.
The theorem then follows by choosing $\varepsilon < e^{-3\alpha}$.

\subsection{Proof of \texorpdfstring{\Cref{thm:hardness-additive}}{Theorem 4.4}}
\label{app:hardness-additive}

Recall that the instance has $m = 3$ outcomes and $n = 3$ actions.
The principal has values $v_0 = 0$ and $v_1 = v_2 = 1$ for outcomes $1$ and $2$.
For a sufficiently small $\varepsilon > 0$, the agent's action costs and the corresponding outcome distributions, omitting the null action and null outcome, are given in the following table.
It is easy to verify that the instance satisfies FOSD and CDFP.

\begin{center}
    \begin{tabular}{ccc}
        \toprule
        \text{\rm Cost} & $F(1|c)$ & $F(2|c)$ \\
        \midrule
        $c_1 = \varepsilon$ & $\nicefrac{1}{2}$ & $\nicefrac{4\varepsilon^2}{H}$ \\
        $c_2 = \nicefrac{1}{4}$ & $1$ & $\nicefrac{\varepsilon}{H}$ \\
        \bottomrule
    \end{tabular}
\end{center}

Consider any contract defined by payments $p_0 = 0$ and $p_1, p_2 \ge 0$.
The agent's utilities for the three actions are $0$ for the null action $0$, and 
\[
    \text{action } 1 ~ \to ~ \frac{1}{2} \cdot p_1 + \frac{4\varepsilon^2}{H} \cdot (p_2 - p_1) - \varepsilon
    ~, \quad
    \text{action } 2 ~ \to ~ 1 \cdot p_1 + \frac{\varepsilon}{H} \cdot (p_2 - p_1) - \frac{1}{4}
    ~.
\]

The optimal contract is $p_1 = 0$ and $p_2 = \nicefrac{H}{4\varepsilon}$.
The agent's utilities are $0$ for all three actions according to the above formulas.
Breaking ties in favor of the principal, the agent takes action $2$.
The principal's utility equals $1 - \nicefrac{1}{4} = \nicefrac{3}{4}$, the 
welfare of this action.

It remains to show that for any contract with payments upper bounded by $H$, the principal's utility is strictly less than $\nicefrac{3}{4} - \beta$ for any $\beta < \nicefrac{1}{4}$.
Consider any contract under which the agent chooses action $1$.
The principal's utility is upper bounded by $\nicefrac{1}{2}$, the 
welfare of this action.
Next, consider a contract that incentivizes action $2$.
To drive the agent to choose action $2$ over action $1$, a contract must satisfy:
\begin{equation}
    \label{eqn:additive-approximation-hardness}
    \frac{1}{2} \cdot p_1 + \frac{\varepsilon}{H} (1 - 4\varepsilon)\cdot (p_2 - p_1) \ge \frac{1}{4}(1 - 4\varepsilon)
    ~.
\end{equation}

Hence, the payment is at least:
\begin{align*}
    1 \cdot p_1 + \frac{\varepsilon}{H} \cdot (p_2 - p_1)
    &
    = 2 \cdot \left( \frac{1}{2} \cdot p_1 + \frac{\varepsilon}{H} (1 - 4\varepsilon)\cdot (p_2 - p_1) \right) - \frac{\varepsilon}{H} (1 - 8\varepsilon) \cdot (p_2 - p_1) \\
    &
    \ge\frac{1}{2}(1 - 4\varepsilon) - \varepsilon (1 - 8\varepsilon) 
    ~,
\end{align*}
where the inequality follows by \Cref{eqn:additive-approximation-hardness} and $0 \le p_1, p_2 \le H$.
Therefore, the principal's utility is at most $1 - \nicefrac{1}{2}(1 - 4\varepsilon) + \varepsilon (1 - 8\varepsilon) \to 1/2$ with $\varepsilon \to 0$. This finishes the proof of the theorem.

\subsection{Proof of \texorpdfstring{\Cref{thm: mixed approximation}}{Theorem 4.5}}
\label{app:mixed-approximation}
	
In the following proofs, we will refer to the linear contract defined by ratio $\rho$ as linear contract $\rho$. For $1 \le i \le k$ where $k = \log \nicefrac{1}{\varepsilon}$, consider linear contract $\rho_i = 1-2^{-i}$.
Let $c_i$ denote the cost of the agent's action incentivized by linear contract $\rho_i$, and $V_i$ be the expected principal's value under that contract. 
For any $i = 1,\dots,k$, by the agent's IC property we have
\[
	\rho_i V_{i+1} -  c_{i+1} \leq \rho_i V_i - c_i < V_i - c_i		
	~,
\]

Therefore, we have \[
	(V_{i+1}-c_{i+1}) - (V_i-c_i) < (1-\rho_i) V_{i+1} = 2(1-\rho_{i+1}) V_{i+1}
	~.
\]

With a telescoping sum over $i = 1\dots, k-1$, we obtain that
\begin{equation*}
	V_k-c_k \leq 2\sum^{k}_{i=2} (1-\rho_i) V_{i} + (V_1-c_1) \leq 2\sum^{k}_{i=2} (1-\rho_i) V_{i} + 2 (1-\rho_1) V_1
	~,  
\end{equation*}
where the last inequality holds because $c_1 \ge 0$ and $\rho_1 = \nicefrac{1}{2}$.
Further, $(1-\rho_i) V_i$ is the principal's utility from linear contract $\rho_i$.
Hence, we get that
\[
    V_k - c_k \leq 2\log\frac{1}{\varepsilon} \,\cdot\, \lin
    ~.
\]

Now consider action $a^* \in \argmax_{a\in A} ~ (V_a-c_a)$. Let $V_{a^*}$ denote the expected principal's value for the outcome distribution of action $a^*$. Since $\OPT$ is upper bounded by $\max_{a\in A} (V_a-c_a)$, it remains to bound 
the differences between $V_{a^*}-c_{a^*}$ and $V_k - c_k$. 
By the agent's IC property for linear contract $\rho_k$, we have $\rho_k V_k - c_k \ge \rho_k V_{a^*} - c_{a^*}$.
Rearranging terms gives
\begin{equation*}
	\left(V_{a^*}-c_{a^*}\right) - \left(V_k - c_k\right) \leq \left(\frac{1}{\rho_k}-1 \right)(c_{a^*}-c_k) \leq \frac{\varepsilon}{1-\varepsilon} \cdot 1 \le 2\varepsilon
	~. 
\end{equation*}
The proof is completed by summing up the above two inequalities. 
 
\bibliographystyle{alphaurl}
\newcommand{\etalchar}[1]{$^{#1}$}

\end{document}